\documentclass[12pt,a4paper]{article}

\usepackage[utf8]{inputenc}
\usepackage[T1]{fontenc}
\usepackage{lmodern}
\usepackage[a4paper, margin=2.5cm]{geometry}
\usepackage{booktabs}
\usepackage{multirow}
\usepackage{longtable}

\usepackage{amsmath,amssymb}
\usepackage{amsthm}
\usepackage{algorithm}
\usepackage{algorithmic}

\usepackage{pgfplots}
\pgfplotsset{compat=1.18}

\usepackage[authoryear,longnamesfirst]{natbib}

\usepackage{xurl}

\usepackage[colorlinks=true,linkcolor=blue,citecolor=blue,urlcolor=blue]{hyperref}

\emergencystretch=\maxdimen
\newtheorem{theorem}{Theorem}[section]
\newtheorem{proposition}{Proposition}[section]
\newtheorem{lemma}{Lemma}[section]
\newtheorem{corollary}{Corollary}[section]
\theoremstyle{definition}
\newtheorem{definition}{Definition}[section]

\theoremstyle{remark}
\newtheorem{remark}{Remark}[section]

\title{\textbf{Complete Exchange Mechanisms}\thanks{We thank Morimitsu Kurino, participants in Japanese Economic Association Meeting 2025 fall, and seminar participants at Osaka Metropolitan University for helpful comments and suggestions. The authors used AI assistance tools (Gemini Pro 2.5 \& Claude Sonnet 4) to help refine the article and develop computational scripts. All AI-generated content was critically reviewed and subsequently edited by the authors. This work was supported by JSPS KAKENHI Grant Number JP21K01391 and 25K05004. All errors are our own.}}

\author{
Minoru Kitahara\thanks{Osaka Metropolitan University. Email: mkitahar [at] omu.ac.jp} \and
Hiroshi Uno\thanks{Osaka Metropolitan University. Email: uno [at] omu.ac.jp}
}

\date{\today}

\begin{document}

\maketitle

\begin{abstract}
\noindent This paper studies one-sided matching under a complete exchange (CE) requirement, where each agent must be assigned an object different from its initial endowment. We introduce assignment partition---a partition of agents and choice sets that builds CE into feasibility---and, within this structure, propose two new mechanisms. 

Chain Serial Dictatorship (C-SD) operates within the partition as a binding-choice chain: the highest-priority agent picks from its allowed set and the right to pick passes to the owner of the chosen object; if that owner has already picked, the right reverts to the highest-priority remaining agent. Two-Stage Serial Dictatorship (T-SD) operates within the partition as a nominate-then-assign procedure: in Stage 1, agents tentatively nominate objects in exogenous priority, and the owners of nominated objects determine an endogenous final priority; in Stage 2, serial dictatorship runs within the partition using that final priority. For any given assignment partition, C-SD and T-SD simultaneously satisfy strategy-proofness, respecting improvement, and efficiency relative to the partition.

We then examine the limits of pursuing market-wide efficiency under the CE constraint. As a benchmark, we study a modified TTC, CE-TTC, which first enforces a CE-compliant reassignment and then runs a self-avoiding top-trading-cycles phase; CE-TTC achieves efficiency within the CE constraint and strategy-proofness but fails respecting improvement. Moreover, for three and four agents, no mechanism can simultaneously achieve efficiency within the CE constraint, respecting improvement, and strategy-proofness. These findings underscore the value of partition-based design for truthful implementation with investment incentives under a hard CE mandate.
\end{abstract}

\vspace{0.1in}
\noindent \textbf{Keywords:} Mechanism design, Market design, One-sided matching, House allocation, Complete exchange, Respecting improvement, Strategy-proofness, Human capital investment, Job rotation, Personnel assignment

\vspace{0.1in}
\noindent \textbf{JEL Classification:} C78, D47, D82

\section{Introduction}\label{sec:introduction}

Organizational objectives and values sometimes necessitate mandatory personnel transfers that override individual preferences. Such systems, where all participants must be reassigned to positions different from their initial assignments, are implemented in limited but important contexts including the periodic rotation of civil servants,\footnote{Periodic transfers are a standard practice in public administration in several countries, including Japan, South Korea, India, and China. The practice is frequently employed as an anti-corruption measure, especially in high-risk sectors like customs agencies \citep{kim2008, bertrandEtAl2020, worldBankGroup2016, bergin2023}.} military personnel assignments,\footnote{Job rotation is a standard practice in military organizations across English-speaking countries, including Australia, Canada, New Zealand, the United Kingdom, and the United States. Officers typically change assignments every two to three years, far more frequently than their civilian counterparts \citep{jansFrazerJans2004}. In Japan, the Self-Defense Forces operate under similar principles, with personnel transfers governed by legal obligations under Article 57 of the Self-Defense Forces Law.} teacher exchanges in Japan and China, and corporate job rotation. These systems are motivated by various organizational imperatives - such as developing human capital, promoting resource equalization, ensuring coverage for less desirable positions, fostering organizational learning, or preventing corruption - that may not be attainable if individual retention rights are upheld \citep{japanInternationalCooperationAgencyHistoryJapanEducational2004, watanabeKenpiFutanKyoshokuin2019, wuNarrowingGapChinese2020, aokiInformationIncentivesBargaining1988, itoJapaneseEconomy2020}.

We formalize this institutional requirement as complete exchange (CE), where every division must receive a worker different from its initial endowment.\footnote{The term ``complete exchange'' has also been studied in the context of a house-swapping problem with complex divisions by \citet{cechlarovaHouseswappingDivorcingEngaged2016}. Their work focuses on the computational complexity of finding an exchange and the necessary conditions for the existence of a complete exchange outcome. In contrast, our paper treats complete exchange as an exogenously imposed institutional constraint and analyzes the design of mechanisms with desirable incentive properties under this rule.} This constraint conflicts with standard notions of efficiency and individual rationality (IR), which ensures voluntary participation \citep{shapleyCoresIndivisibility1974}. When every division prefers its own worker, any CE assignment is Pareto dominated by the initial endowment, making standard efficient mechanisms like serial dictatorship and top trading cycles inapplicable.

This conflict creates an incentive problem: divisions may underinvest in worker development when they cannot retain the benefits of their investments \citep{beckerHumanCapitalTheoretical1964, pigouWealthWelfare1912}. To address this, mechanisms should satisfy respecting improvement (RI), which ensures that when a division's worker becomes more attractive to others, the division's assignment does not worsen \citep{biroShapleyScarfHousing2024}. 

To ensure complete exchange is met, one practical approach is to impose structural constraints on the assignment process. Additionally, any practical mechanism must be strategy-proof (SP), ensuring truthful preference revelation \citep{satterthwaiteStrategyproofAllocationMechanisms1981}. This leads to the central research question of our paper:
\begin{itemize}
\item
\emph{What complete exchange mechanism achieves some efficiency under strategy-proofness and respecting improvement?}
\end{itemize}
To answer this question, we propose two novel mechanisms, Chain Serial Dictatorship (C-SD) and Two-Stage Serial Dictatorship (T-SD). To make our approach concrete, we first illustrate the structure and the mechanisms through a simple example. 

Consider an assignment problem with six divisions and six workers, where each division initially owns one worker and has strict preferences over all workers. Each division initially owns a worker with the same name ($A_1$ owns $a_1$, $A_2$ owns $a_2$, etc.). We also have an exogenous priority order over divisions. We partition both the divisions and workers into two equal-sized groups: Group A consists of divisions $\{A_1, A_2, A_3\}$ and the choice set of workers $\{b_1, b_2, b_3\}$, and Group B consists of divisions $\{B_1, B_2, B_3\}$ and the choice set of workers $\{a_1, a_2, a_3\}$. 

This partitioning creates what we call an \textit{assignment partition}, which consists of pairs of division groups and their corresponding worker choice sets. Each pair satisfies two key properties: (1) \textit{separation}, meaning no division in a group initially owns any worker in that group's choice set, and (2) \textit{balance}, meaning each group has the same number of divisions and workers. 

These properties ensure that each division can only choose workers from the other group, which automatically guarantees complete exchange. For instance, $A_1$ can only choose from $\{b_1, b_2, b_3\}$ (Group B workers), and $B_1$ can only choose from $\{a_1, a_2, a_3\}$ (Group A workers).

Assume that all divisions have identical preferences: $a_2 \succ a_1 \succ a_3 \succ b_1 \succ b_2 \succ b_3$. The exogenous priority order is $A_1 \triangleright A_2 \triangleright A_3 \triangleright B_1 \triangleright B_2 \triangleright B_3$.

\paragraph{Chain Serial Dictatorship (C-SD)}

C-SD implements a chain-like process where the right to choose next passes to the division whose worker was just selected. 
The process begins with the division that has the highest priority according to the exogenous priority order. Since $A_1$ has the highest priority in the order $A_1 \triangleright A_2 \triangleright A_3 \triangleright B_1 \triangleright B_2 \triangleright B_3$, $A_1$ gets to choose first.

The process proceeds as follows:
\begin{enumerate}
\item $A_1$ chooses first and selects $b_1$, then exits. The right to choose next passes to $B_1$ (the original owner of $b_1$).
\item $B_1$ selects $a_2$ and exits. The right to choose next passes to $A_2$ (the original owner of $a_2$).
\item $A_2$ selects $b_2$ (since $b_1$ is already taken) and exits. The right to choose next passes to $B_2$ (the original owner of $b_2$).
\item $B_2$ selects $a_1$ and exits. Since $A_1$ (the original owner of $a_1$) has already exited, the chain breaks and the remaining division with highest priority, $A_3$, gets to choose next.
\item $A_3$ selects $b_3$ and exits. The right to choose next passes to $B_3$ (the original owner of $b_3$).
\item $B_3$ selects $a_3$ and exits.
\end{enumerate}

The final assignment is: $A_1$ gets $b_1$, $A_2$ gets $b_2$, $A_3$ gets $b_3$, $B_1$ gets $a_2$, $B_2$ gets $a_1$, and $B_3$ gets $a_3$.

\paragraph{Two-Stage Serial Dictatorship (T-SD)}

T-SD operates in two stages. In Stage 1, divisions make tentative choices to determine the selection order for Stage 2. In Stage 2, divisions make final choices based on this order.

Stage 1:
\begin{itemize}
\item Tentative Choices: Each division tentatively chooses its most-preferred worker from the other group, following a predetermined order within each group. Group A goes first: $A_1$ tentatively chooses $b_1$, $A_2$ tentatively chooses $b_2$, $A_3$ tentatively chooses $b_3$. Then Group B: $B_1$ tentatively chooses $a_2$, $B_2$ tentatively chooses $a_1$, $B_3$ tentatively chooses $a_3$.

\item Order Determination: The selection order for Stage 2 is determined by the order in which workers were tentatively chosen. A division gets to choose earlier if its worker was tentatively chosen earlier by the other group. This results in Group A's selection order: $A_2 \rightarrow A_1 \rightarrow A_3$, and Group B's selection order: $B_1 \rightarrow B_2 \rightarrow B_3$.
\end{itemize}

Stage 2 (Final Choices): Group A chooses first in the determined order ($A_2 \rightarrow A_1 \rightarrow A_3$): $A_2$ chooses $b_2$, $A_1$ chooses $b_1$, $A_3$ chooses $b_3$. Then Group B chooses in their determined order ($B_1 \rightarrow B_2 \rightarrow B_3$): $B_1$ chooses $a_2$, $B_2$ chooses $a_1$, $B_3$ chooses $a_3$.

The final assignment is identical to C-SD: $A_1$ gets $b_1$, $A_2$ gets $b_2$, $A_3$ gets $b_3$, $B_1$ gets $a_2$, $B_2$ gets $a_1$, and $B_3$ gets $a_3$.

The paper's first main contribution is to show that both C-SD and T-SD provide a positive answer to our research question. They are strategy-proof, respect improvement, and are efficient within the constraints of the partition. We formalize this latter property as efficiency under assignment partition (EAP). These mechanisms offer a robust solution for practical implementation, as we illustrate with a proposed reform for the NPB Active Player Draft, which can leverage its two-league structure as a natural assignment partition. 

We illustrate our framework's policy relevance by analyzing the ``Active Player Draft'' (\emph{gen'eki dorafuto}) in Nippon Professional Baseball (NPB). This draft mandates complete exchange to create new player opportunities and implicitly aims for respecting improvement to encourage clubs to list valuable players. While the mechanism satisfies CE-efficiency (efficiency under the complete exchange constraint) and strategy-proofness for $n=3$ \citep{umenishiGenekiDorafutoRuuru2025}, we prove that its desirable properties do not extend to larger markets: CE-efficiency is maintained for all $n\ge 3$, but strategy-proofness fails for $n\ge 4$ (generalizing the $n=4$ result from \citet{umenishiGenekiDorafutoRuuru2025}), and respecting improvement---not examined in the $n=3$ case---fails already at $n\ge 3$. We propose a reform based on our C-SD and T-SD mechanisms, using the NPB's two-league structure as a natural assignment partition. This reform would not only satisfy complete exchange, strategy-proofness, and respecting improvement, but also achieve efficiency {under the restriction of cross-league transfers} and hence resolve practical issues with same-league transfers, thus contributing to the draft's stated goals.

Our practical approach of using a partition, however, raises a question: why rely on such a seemingly restrictive structure at all? An alternative approach would be to pursue efficiency under the complete exchange constraint directly, without imposing any specific partition. This notion of efficiency, which we term CE-efficiency, considers efficiency from the perspective of allocations restricted to CE, aiming for overall efficiency within the CE constraint rather than the efficiency achievable under assignment partition. The pursuit of such a mechanism forms the second arc of our paper.

To answer this question, we establish impossibility results that demonstrate the fundamental constraints on achieving CE-efficiency while preserving key incentive properties. For small markets ($n=3, 4$), we prove that CE-efficiency and respecting improvement are incompatible---any mechanism that achieves CE-efficiency must violate RI, creating a fundamental trade-off between efficiency and investment incentives. For $n=4$, while CE-efficiency and respecting improvement can be simultaneously satisfied, adding strategy-proofness creates impossibility. These impossibility results establish boundaries that demonstrate the fundamental constraints on achieving CE-efficiency while preserving key incentive properties.

\paragraph{Our Contribution}\label{ssubsec:our-contribution}

This paper makes a three-fold contribution to one-sided matching and market design under mandatory complete exchange.

First, we provide the first formal analysis of CE mechanisms as a primary design constraint in one-sided matching. While CE conflicts with standard efficiency and individual rationality assumptions, many institutions require it. We treat CE as a central design challenge rather than an obstacle.

Second, we identify two critical incentive problems under mandatory exchange: (a) maintaining organizational investment incentives (respecting improvement) and (b) achieving efficient allocation based on true preferences, which requires ensuring truthful preference revelation (strategy-proofness) as a means. We show these can be simultaneously satisfied through structured design.

Third, we propose two mechanisms---Chain Serial Dictatorship and Two-Stage Serial Dictatorship---that simultaneously achieve strategy-proofness, respecting improvement, and efficiency under assignment partition.

We demonstrate practical relevance through a proposed reform of the NPB Active Player Draft, showing how our framework resolves strategic flaws while advancing institutional objectives.

\paragraph{Limitations and Scope}

Our research addresses a specific technical question: given that CE is exogenously mandated, how can we design mechanisms that address two key incentive problems (under-investment and misreporting) while achieving partial efficiency? CE requirements force us to relinquish Pareto efficiency and individual rationality, but we do not evaluate whether the benefits of CE justify these costs, nor do we analyze CE's broader consequences or overall desirability.

We analyze the CE assignment problem as one-sided matching, considering only division preferences. While two-sided matching incorporating worker preferences could potentially yield better CE allocations, and the choice between one-sided and two-sided approaches presents important organizational design questions, we do not address these issues. For empirical analysis of this choice, see \citet{cowgillStableMatchingJob2025a}.

Our analysis focuses on one-shot assignment problems. In practice, personnel assignments involve long-term considerations, and improvement incentives operate over extended periods. While CE violates IR constraints in one-shot settings, job rotation systems may offer solutions that satisfy IR for most divisions in the long run. This represents an important direction for future research.

We model respecting improvement as a property that prevents divisions from being harmed when their workers improve, assuming these improvements can be achieved without cost. In reality, improvements often require costly investments. We do not address optimal investment levels or the dynamic aspects of improvement decisions.

Additional technical limitations include: static analysis excluding externalities, complementarities, team production, and inter-divisional competition; deterministic mechanisms only, leaving randomization and group strategy-proofness for future work; exogenous assignment partition without optimization. Our impossibility results show hierarchical structure: for $n=3$, CE-efficiency and respecting improvement are incompatible; for $n=4$, adding strategy-proofness creates incompatibility. For larger $n$, the feasibility of CE-efficiency with all incentive properties remains open. Finally, while our impossibility results establish boundaries for CE-efficiency, they do not establish that the partition approach is necessary---the space between efficiency under assignment partition and CE-efficiency remains largely unexplored, representing an important direction for future research. 

\paragraph{Related Literature}

Our research is positioned at the intersection of matching theory, particularly work on investment incentives, and the economics of personnel practices like job rotation. We first review the foundations that our model builds upon and departs from, then survey the institutional context that motivates our inquiry.

A common assumption in matching theory, tracing back to \citet{shapleyCoresIndivisibility1974}, is individual rationality. This principle guarantees that no participant is made worse off than their initial endowment, serving as a participation constraint in voluntary markets. In practice, it ensures participants in a housing allocation are guaranteed a home at least as good as their initial one \citep{abdulkadirogluHouseAllocationExisting1999} and that patients in a kidney exchange are not harmed by participating \citep{rothKidneyExchange2004}.

Our work explores environments where this assumption is deliberately violated by a central planner through a mechanism of complete exchange, where all divisions are reassigned. While this may seem counterintuitive, it is often motivated by long-term organizational goals that supersede short-term individual preferences. For instance, in mandatory job rotation systems, employees may be assigned to temporarily undesirable but operationally necessary posts to develop skills or prevent corruption, under the implicit guarantee that such sacrifices are part of a broader, equitable career trajectory \citep{ortegaJobRotationLearning2001, campionEtAl1994}. While we do not explicitly model such dynamic considerations, this concept of ``long-term IR''---where expected utility over a career path is positive, even if short-term assignments are not---provides one possible justification for CE mechanisms. This contrasts with models focused on improving upon existing matches, which implicitly respect individual rationality.

The closest precursor to our work is the model of job rotation by \citet{yuMarketDesignApproach2020}, who also study a setting where individual rationality is not guaranteed assignment-by-assignment. However, their focus is on the cyclical nature of rotation, whereas our model accommodates reassignments within assignment partition structures and centers on the incentive implications of respecting divisions' endogenous investments. The respecting improvement property has been studied in various matching contexts \citep{balinskiTaleTwoMechanisms1999, hatfieldImprovingSchoolsSchool2016}, and the incentive problems of talent poaching and multiskilling have been analyzed in organizational contexts \citep{battistonTalentPoachingJob, carmichaelMultiskillingTechnicalChange1993}. Our work contributes to this literature by providing the first formal analysis of respecting improvement under a strict complete exchange constraint, bridging the literature on investment incentives (which assumes IR) with the design challenges inherent in mandatory exchange systems. We introduce a framework for implementing complete exchange through assignment partition, establishing necessary and sufficient conditions for its existence and providing efficient algorithms for its construction.

\paragraph{Institutional Context}\label{ssubsec:institutional_context}

CE mechanisms, though rarely labeled as such, are prevalent across various sectors. They are typically implemented to address systemic issues like corruption, to foster skill development, or to ensure that unpopular but necessary positions are filled.

\subparagraph{Professional Sports Drafts}\label{ssubsec:sports}

While not a perfect analogy, certain features of professional sports leagues resonate with CE principles. In Japan's Nippon Professional Baseball (NPB), the ``Active Player Draft'' (Gen'eki Draft) functions as a limited CE system \citep{nikkanSports2022}. Teams must submit a list of players they are willing to trade, and a centralized process reallocates a subset of these players to other teams. The goal is to provide playing opportunities for athletes who are stuck on the depth chart of their current team, thereby increasing competitive balance and player opportunities. This mechanism, though limited in scope, demonstrates how complete exchange can be adapted to facilitate movement and opportunity in a closed market.

\subparagraph{Civil Servant Assignment}\label{ssubsec:civil_servant}

Mandatory rotation is a feature of civil service personnel policy in numerous countries, serving diverse and sometimes conflicting objectives, from anti-corruption to strategic leadership development \citep{worldBankGroup2016}. The design of these systems often embodies a CE constraint. In South Korea, for instance, a high-frequency ``cyclical reshuffle'' system mandates that senior civil servants rotate positions, often annually, to prevent corruption and foster a generalist perspective \citep{kim2008}. Similarly, India's elite Indian Administrative Service (IAS) has historically used centralized, tenure-based rotation policies, although recent reforms have attempted to incorporate officer preferences \citep{bertrandEtAl2020}.

These systems, however, are not without costs. Frequent rotations can lead to a loss of specialized expertise and lower morale, as observed in customs administrations where it may disrupt efforts to combat smuggling \citep{bergin2023}. While some systems are rigid, others employ more targeted approaches. The U.S. federal government uses ``directed reassignments'' for specific strategic purposes \citep{officeOfPersonnelManagement2016}, and countries like Germany and Italy apply risk-based rotations in corruption-prone agencies such as the tax authority \citep{worldBankGroup2016}. Our model provides a framework for analyzing the trade-offs inherent in these diverse implementations.

\subparagraph{Military Personnel Assignment}\label{ssubsec:military}

Military organizations worldwide rely heavily on CE-type systems to manage personnel. The U.S. Department of Defense, for example, operates a centralized assignment system where service members are regularly rotated through different posts to meet operational needs, facilitate career progression, and develop a broad range of skills \citep{uSDepartmentOfDefense2015}. Officers, in particular, are moved between operational, training, and administrative roles to cultivate the leadership qualities required at senior levels. Individual preferences may be considered, but the needs of the service are paramount, making this a real-world example of a centrally mandated CE mechanism.

\paragraph{Corporate Job Rotation: Organizational Imperatives}\label{ssubsec:corporate_rotation}

In corporate settings, mandatory rotation serves distinct organizational objectives that complement but differ from educational contexts. Rotation transfers tacit knowledge, breaks down silos, and develops multi-skilled employees. \citet{ortegaJobRotationLearning2001} provides a foundational analysis of rotation as a learning mechanism, and recent work studies dynamic trade-offs between rotation and specialization using matching models \citep{kurinoJobRotationSpecialization2024}. 

Rotation also addresses incentive challenges. Experimentally, job rotation can eliminate the ratchet effect in compensation \citep{weiCanJobRotation2020}, and it can reinvigorate motivation by resetting career concerns \citep{hakenesIncentiveEffectsJob2017}. In Japanese firms, rotation underpins broad skill development and organizational knowledge \citep{aokiInformationIncentivesBargaining1988}. These studies, however, focus on \emph{employee} incentives, whereas our research targets \emph{organizational} investment incentives under mandatory exchange.

Applications extend beyond human capital development. Evidence from internal talent markets suggests that division-led assignment can enhance productivity relative to worker-preference-based systems \citep{cowgillStableMatchingJob2025a}, supporting our division-centric modeling. Rotation also serves governance functions, including anti-corruption.

Despite benefits, rotation has costs: mandatory loan officer rotation in Indian banks induced temporary efficiency losses during adaptation \citep{bhowalCostsJobRotation2021}. Thus, CE is a prevalent institutional constraint motivated by organizational goals, but it entails trade-offs between knowledge transfer and job-specific human capital destruction. Our framework addresses how to design mechanisms that preserve \emph{organizational} investment incentives when rotation is exogenously mandated.

\section{Assignment Problems}\label{sec:model}

Consider \(n\) divisions, where each division initially has one worker. Each division has strict preferences over all workers. This is a one-sided matching environment (objects have no preferences). This situation naturally arises in various contexts, such as mandatory teacher reassignments across schools or compulsory employee rotations across departments, where workers cannot express preferences.

We formalize this situation as follows:

\begin{definition}[Assignment Problem]\label{def:assignment}
An \textbf{assignment problem} is a tuple $\mathcal{A}=(N, X, \succ)$, where:
\begin{itemize}
\item $N=\{1,2,\ldots,n\}$ is the set of divisions,
\item $X$ is the set of workers, with $|X|=|N|$,
\item $\succ:=(\succ_i)_{i \in N}$ is a preference profile, where $\succ_i$ is the strict preference of division $i$ over $X$.
\end{itemize}
\end{definition}

We denote by $\mathcal{P}$ the set of all preference profiles.

For any division $i$, we denote by $\succeq_i$ the weak preference relation induced by $\succ_i$: $j \succeq_i h$ if and only if $j \succ_i h$ or $j = h$.

When a strict priority order $\triangleright$ over divisions is given, we call it an \textbf{assignment problem with priority} and denote it as $\mathcal{A}=(N, X, \succ, \triangleright)$.

For notational convenience, and without loss of generality, we assume that division $i$ initially owns worker $i$. This assumption implies that the set of workers $X$ equals the set of divisions $N$ (i.e., $X=N$). 

Let $o:N\to X$ be the identity owner map, $o(i)=i$ for all $i\in N$.

An assignment specifies which worker each division receives. Since we have an equal number of divisions and workers, we consider complete assignments where every division receives exactly one worker and every worker is assigned to exactly one division.

\begin{definition}[Assignment]\label{def:assignment-mu}
An \textbf{assignment} is a bijective function $\mu: N \to X$, where $\mu_i:=\mu(i)$ denotes the worker assigned to division $i$ under $\mu$.
\end{definition}

We denote by $\mathcal{M}$ the set of all assignments.

A mechanism is a systematic way to determine assignments based on divisions' preferences. It takes as input the preference profile of all divisions and produces an assignment as output.

\begin{definition}[Mechanism]\label{def:mechanism}
A \textbf{mechanism} is a function $f: \mathcal{P} \to \mathcal{M}$ that assigns to each preference profile $\succ \in \mathcal{P}$ an assignment $\mu \in \mathcal{M}$.
\end{definition}

We write $f_i(\succ)$ to denote the worker assigned to division $i$ under mechanism $f$ at preference profile $\succ$.

\subsection{Complete Exchange, Respecting Improvement, and Standard Properties}\label{subsec:properties}

In most allocation problems, it is natural to require that no division should be worse off by participating in the mechanism than they would be with their initial endowment. This requirement, called individual rationality, ensures voluntary participation.

However, some institutional settings require all participants to change their workers. Such requirements arise when individual rationality fails to secure adequate coverage of essential positions. For instance, rural schools may struggle to attract teachers voluntarily, necessitating mandatory rotation \citep{japanInternationalCooperationAgencyHistoryJapanEducational2004}. Similarly, less desirable corporate positions may remain unfilled under voluntary systems, requiring mandatory rotation. These constraints demonstrate that complete exchange is often an institutional necessity rather than merely an organizational preference.

Complete exchange is a property in assignment mechanisms that ensures every division must receive a different worker from its initial one. This property is used in institutional settings where rotation or exchange is mandatory, such as in teacher rotation systems or corporate job transfers.

\begin{definition}[Complete Exchange (CE)]\label{def:CE}
A mechanism $f$ satisfies \textbf{complete exchange (CE)} if for any preference profile $\succ$, $f_i(\succ) \neq i$ for all $i \in N$.
\end{definition}

We denote by $\mathcal{C}$ the set of all assignments that satisfy complete exchange, i.e., the set of all derangements (permutations with no fixed points) on $N$.

It is immediate from their definitions that individual rationality and complete exchange are conflicting properties. To see this, consider a preference profile where every division regards its own initial worker as the most desirable. Under such preferences, any assignment that satisfies complete exchange necessarily makes every division worse off, thus violating individual rationality. Conversely, any assignment that satisfies individual rationality must allow each division to keep its worker, thus violating complete exchange.

While complete exchange ensures the desired rotation outcome, it may create a potential incentive problem for divisions. Since divisions must exchange their workers and cannot retain them regardless of the investment made in worker development, division managers may lack sufficient motivation to invest in improving their workers' skills and capabilities. This is because the benefits of such investments---improved worker quality---would accrue to other divisions that receive the enhanced workers, while the investing division bears the full cost of development without retaining the improved worker.

To address this incentive misalignment, we introduce the respecting improvement property. Respecting improvement \citep{biroShapleyScarfHousing2024} ensures divisions have proper incentives to improve their workers' quality. When a division's worker's ranking improves relative to other workers in the preferences of other divisions, respecting improvement guarantees that the division's assignment will not worsen. 

\begin{definition}[Improvement and Respecting Improvement (RI)]\label{def:RI}
A preference profile $\succ'$ is said to be an \textbf{improvement for division} $i \in N$ \textbf{with respect to} $\succ$ if:
\begin{enumerate}
\item $\succ'_i = \succ_i$ (division $i$'s preferences remain unchanged)
\item for all $j \neq i$ and all $k \neq i$, if $i \succ_j k$, then $i \succ'_j k$ (worker $i$'s ranking can only improve)
\item for all $j \neq i$ and all $k, l \neq i$, $k \succ_j l \Leftrightarrow k \succ'_j l$ (relative rankings of other workers remain unchanged)
\end{enumerate}

A mechanism $f$ satisfies \textbf{respecting improvement (RI)} if, for each division $i \in N$ and any preferences $\succ, \succ'$ such that $\succ'$ is an improvement for $i$ with respect to $\succ$, we have $f_i(\succ') \succeq_i f_i(\succ)$.
\end{definition}

Strategy-proofness ensures that divisions cannot benefit from misreporting their preferences. This property maintains the integrity and efficiency of the mechanism, as it guarantees that the mechanism's outcomes are based on divisions' true preferences rather than strategic misrepresentation.

\begin{definition}[Strategy-proofness (SP)]\label{def:SP}
A mechanism $f$ satisfies \textbf{strategy-proofness (SP)} if for any division $i \in N$, any preference profile $\succ$, and any alternative preference $\succ'_i$ of division $i$, we have $f_i(\succ) \succeq_i f_i(\succ'_i, \succ_{-i})$.
\end{definition}

Efficiency is a criterion in evaluating mechanism properties, as it ensures that there is no waste in the assignment outcome.

An assignment is Pareto efficient if there exists no other assignment that makes at least one division better off without making anyone else worse off.

Like individual rationality, Pareto efficiency can be in direct conflict with the complete exchange constraint. The same example profile illustrates the point: when every division's most-preferred worker is its own, the only Pareto efficient assignment is the initial endowment itself. Any other assignment would make every division strictly worse and would thus be Pareto dominated by the initial state. Consequently, for this profile, any mechanism that satisfies complete exchange must produce a Pareto inefficient outcome, while any Pareto efficient mechanism must violate complete exchange.

When mechanisms operate under assignment partition structures, efficiency should be evaluated within the constraints imposed by this partition. The question of whether complete exchange, strategy-proofness, respecting improvement, and general complete exchange-constrained efficiency can be simultaneously achieved remains an open problem. However, within the specific structure of assignment partition, we can define and achieve a meaningful efficiency concept: efficiency under assignment partition, which we will formally define after introducing our proposed mechanisms.

Across educational and corporate contexts, mandatory personnel exchange serves institutional objectives but creates a common challenge: how to maintain investment incentives when the benefits of development accrue to other units. This challenge transcends sectors and poses a mechanism design problem. Our framework focuses on this problem and develops mechanisms that sustain investment incentives when CE is institutionally required.

\section{Complete Exchange Mechanisms}\label{sec:mechanisms}

We study two mechanisms that achieve complete exchange, strategy-proofness, and respecting improvement: Chain Serial Dictatorship and Two-Stage Serial Dictatorship. 

Both mechanisms operate within a structural constraint we call an \textbf{assignment partition}. Intuitively, we partition divisions and workers into balanced groups so that no division can select its incumbent worker. Given this structural guarantee of complete exchange, the mechanisms build upon the Serial Dictatorship (SD) framework---where divisions sequentially select their most-preferred available option according to a priority order---but differ in how they endogenously determine these selection priorities:

\begin{itemize}
\item \textbf{Chain Serial Dictatorship} iteratively determines the final order through a chain of selections: the right to choose passes to the owner of the selected object, creating an adaptive sequence similar to a modified SD.
\item \textbf{Two-
stage Serial Dictatorship} uses two rounds of SD-like processes: in Stage 1, divisions make tentative nominations in an exogenous order to determine a final order; in Stage 2, a final SD is run using this final order to make binding assignments.
\end{itemize}

This SD-based structure ensures that, like traditional Serial Dictatorship, selections are made sequentially and bindingly, but with priorities adapted to incentivize truthful behavior and respect improvements under the CE constraint.

We begin by defining assignment partitions and stating a general efficiency notion relative to a given partition. We then describe Chain Serial Dictatorship and Two-
stage Serial Dictatorship.

\subsection{Assignment Partition}\label{subsec:a_partition}

To guarantee complete exchange structurally, we employ an assignment partition that partitions divisions and workers into groups such that no division can select a worker from its own group. This separation ensures that each division must choose from workers initially assigned to other groups, thereby eliminating the possibility of retaining one's own worker.

\begin{definition}[Assignment Partition]\label{def:a_partition}
Let $\mathcal{A}=(N, X, \succ)$ be an assignment problem. An \textbf{assignment partition} of $\mathcal{A}$ is a collection $\mathcal{D} = \left( (N_k, X_k) \right)_{k=1}^K$ of pairs such that:
\begin{itemize}
\item $\{N_k\}_{k=1}^K$ partitions the divisions $N$, $\{X_k\}_{k=1}^K$ partitions the workers $X$, and
\item for each group $k$: (i) $N_k \cap X_k = \emptyset$ (separation), and (ii) $|N_k| = |X_k|$ (balance).
\end{itemize}
\end{definition}

For notational convenience, we denote $n_k = |N_k|$ for each group $k$, and let $g(i)$ denote the group index for division $i \in N$.

An assignment partition can always be constructed when the grouping of divisions can be chosen freely. We illustrate this with two cases depending on the parity of the total number of divisions, $n$.

\begin{enumerate}
\item
  If $n=2k$ is even: Partition $N$ into two equal groups $N_1, N_2$ of size $k$ each, and set $X_1=N_2$ and $X_2=N_1$.
\item
  If $n=2k+1$ is odd: Partition $N$ into three groups with sizes $1, k, k$. Take any worker $j$ from the second group to form the worker sets appropriately.
\end{enumerate}

In practice, divisions are often partitioned by geographical or institutional constraints. For such predetermined partitions, an assignment partition exists if and only if no group contains more than half the total divisions.

\begin{proposition}\label{prop:a_partition-existence}
Let $N$ be a set of $n$ divisions. For any given partition $\{N_k\}_{k=1}^K$ of $N$ with $K \ge 2$, there exists a partition $(X_k)_{k=1}^K$ of $N$ such that $(N_k, X_k)_{k=1}^K$ is an assignment partition if and only if
$$\max_{k \in \{1,\dots,K\}} n_k \le \frac{n}{2}.$$
\end{proposition}

\begin{proof}[Proof of Necessity]
Suppose an assignment partition exists. For any group $k$, the separation property requires $X_k \subseteq N \setminus N_k$, and the balance property requires $|X_k| = n_k$. Combining these, we get
$$n_k = |X_k| \le |N \setminus N_k| = n - n_k,$$
which simplifies to $2n_k \le n$. This must hold for all $k$, thus $\max_k n_k \le n/2$.
\end{proof}

While the full equivalence can be established using Hall's Marriage Theorem, we provide a direct and constructive proof.\footnote{Indeed, Hall's marriage problem can be solved by algorithms designed to find a maximum matching in a bipartite graph. Here, $V$ denotes the number of vertices and $E$ denotes the number of edges in the bipartite graph. Since our condition guarantees that a perfect matching exists, any such algorithm will return a perfect matching. However, applying a general-purpose algorithm results in a super-linear time complexity. For example, the Hopcroft-Karp algorithm finds a maximum matching in $O(E\sqrt{V})$ time \citep{hopcroftN52Algorithm1973}. In our setting, with $V=2n$ vertices and the number of edges $E$ potentially being $O(n^2)$, this complexity becomes $O(n^{2.5})$.} General-purpose algorithms for Hall's marriage problem achieve super-linear time complexity, but our constructive approach achieves linear time complexity.

The core idea is to process groups in order of size, using the largest group as a buffer to ensure no division gets its own worker. The process is as follows:
\begin{enumerate}
\item
  Find the largest group among all division groups.
\item
  Reindex groups so the largest group becomes ``Group 1'' and others follow in processing order.
\item
  Create a worker queue by placing all workers from Group 1 at the front, followed by workers from other groups.
\item
  Distribute sequentially: Each smaller group (Groups 2, 3, ...) takes workers from the front of the queue until it gets enough workers.
\item
  Give the remainder to the largest group: What's left goes to Group 1.
\end{enumerate}

The algorithm's correctness relies on the buffer strategy: by placing the largest group's workers at the front of the queue, they are consumed first by smaller groups. Since the total size of smaller groups is at least as large as the largest group (guaranteed by $\max_k n_k \le n/2$), they collectively absorb all workers from the largest group. When the largest group finally gets its assignment, only workers from other groups remain available, ensuring the separation property $N_k \cap X_k = \emptyset$ holds for all groups.

The formal proof of sufficiency is presented in Appendix \ref{app:a_partition-sufficiency}.

\begin{remark}\label{rem:a_partition-examples}
The examples presented after the definition of an assignment partition (partitioning into two or three groups) serve a dual purpose. Initially, they illustrate that a partition always exists when the grouping is not given. In the context of Proposition \ref{prop:a_partition-existence}, they represent a class of coarsest possible partitions. They are notable for satisfying the condition $\max_k n_k \le n/2$ with the minimum number of groups, effectively pushing the group sizes to their maximum feasible size. The case of an even $n$ is particularly illustrative as it meets the condition with equality, representing a boundary case for the existence of any assignment partition.
\end{remark}

Under such an assignment partition, we define efficiency and complete exchange mechanisms.

\subsection{Efficiency under Assignment Partition (EAP)}\label{subsec:efficiency-under-a_partition}

With the assignment partition structure in place, we define an efficiency concept. Since standard Pareto efficiency conflicts with complete exchange, we evaluate outcomes relative to the set of assignments permitted by a given partition $\mathcal{D}$.

For $S\subseteq N$, write $\mu(S):=\{\mu_i: i\in S\}$ for the image of $S$ under $\mu$.

\begin{definition}[Efficiency under Assignment Partition (EAP)]\label{def:eap}
Let $\mathcal{D} = (N_k, X_k)_{k=1}^K$ be a given assignment partition. An assignment $\mu$ is \textbf{efficient under $\mathcal{D}$} if it is feasible under $\mathcal{D}$ (i.e., $\mu(N_k) = X_k$ for all $k$) and there exists no other feasible assignment $\mu'$ under $\mathcal{D}$ such that $\mu'_i \succeq_i \mu_i$ for all $i \in N$ and $\mu'_j \succ_j \mu_j$ for some $j \in N$.

A mechanism $f$ satisfies \textbf{efficiency under assignment partition} $\mathcal{D}$ if for any preference profile $\succ$, the assignment $f(\succ)$ is efficient under $\mathcal{D}$.
\end{definition}

This definition establishes a hierarchy of efficiency criteria based on the coarseness of the assignment partition. A coarser partition (with larger, fewer groups) permits more potential trades, making the associated efficiency standard more demanding. Conversely, a finer partition makes the efficiency criterion weaker. For example, under the finest partition where each division forms its own group, any assignment is trivially efficient, as no trades are possible within any group.

\subsection{Chain Serial Dictatorship}\label{subsec:CSD}

Given a predetermined assignment partition, the Chain Serial Dictatorship mechanism operates through the following sequential process:

\begin{itemize}
\item The highest-priority division according to the exogenous priority order $\triangleright$ begins the process as the first active division.
\item At each step, the active division selects its most preferred worker from the set of workers assigned to its group, and is immediately assigned that worker. The chosen worker is then removed from the available pool.
\item The selection right passes to the original owner of the just-chosen worker---this chaining feature gives the mechanism its name.
\item If the original owner has already made its selection, the right instead passes to the highest-priority division among those that have not yet chosen.
\end{itemize}

The process continues until every division has selected exactly one worker.

\begin{definition}[Chain Serial Dictatorship]\label{def:CSD}
Let $\mathcal{A}=(N, X, \succ, \triangleright)$ be an assignment problem, and let $\left( N_k, X_k \right)_{k=1}^K$ be an assignment partition of $\mathcal{A}$. The \textbf{Chain Serial Dictatorship mechanism} $f^{C}: \mathcal{P} \to \mathcal{M}$ is defined by the following sequential process:

Initialization: Set $i_1 = \max_{\triangleright} N$ and $W_0 = \emptyset$.

For each step $t = 1, 2, \ldots, n$:
\begin{enumerate}
\item Selection: Division $i_t$ chooses worker:
   $w_t = \max_{\succ_{i_t}}(X_{g(i_t)} \setminus W_{t-1})$
\item Assignment: Set $f^{C}_{i_t}(\succ) = w_t$ and update $W_t = W_{t-1} \cup \{w_t\}$.
\item Next selector (if $t < n$): 
   $i_{t+1} = \begin{cases}
   o(w_t) & \text{if } o(w_t) \notin \{i_1, \ldots, i_t\} \\
   \max_{\triangleright}(N \setminus \{i_1, \ldots, i_t\}) & \text{otherwise}
   \end{cases}$
\end{enumerate}

The mechanism operates through two types of transitions:
\begin{itemize}
\item \textbf{owner-call}: When $o(w_t) \notin \{i_1, \ldots, i_t\}$, the right to choose passes to the owner of the selected worker, creating a chain of selections.
\item \textbf{fallback}: When $o(w_t) \in \{i_1, \ldots, i_t\}$ (i.e., the owner has already moved), the next selector is determined by the exogenous priority $\triangleright$ among remaining divisions.
\end{itemize}
\end{definition}

The assignment partition structure guarantees that no division can retain its own worker, since each division's worker belongs to a different group than the division itself (separation property). This ensures complete exchange is achieved by construction.

\begin{remark}[Conceptual Distinction from YRMH-IGYT]\label{rem:YRMH-distinction}
The chaining rule in C-SD appears similar to the ``You Request My House, I Get Your Turn'' principle of YRMH-IGYT. However, this surface-level similarity conceals a difference in their underlying procedural approaches.

YRMH-IGYT operates on tentative pointings. A division's ``choice'' is a declaration of intent, and assignments are finalized only when a cycle of these pointings is formed and cleared. This logic of resolving trades via cycles makes YRMH-IGYT a procedural variant of Top Trading Cycles.

In contrast, C-SD operates on final selections. Each choice is an irrevocable, binding assignment that removes the selected worker from the pool. The ``turn'' that is passed is not for a tentative pointing but for the next binding assignment in the sequence. This logic of creating a final assignment via an ordered sequence of binding choices makes C-SD a sequential variant of Serial Dictatorship.

This distinction clarifies their intellectual heritage: while both mechanisms use a ``request-and-get-turn'' process, one points toward the exchange-based logic of TTC, and the other toward the priority-based logic of SD.
\end{remark}

\subsection{Two-Stage Serial Dictatorship (T-SD)}\label{subsec:TSD}

Given a predetermined assignment partition, the Two-Stage Serial Dictatorship mechanism operates as follows:

\begin{itemize}
\item Stage 1: Run serial dictatorship under the assignment partition using exogenous priority $\triangleright$ to determine tentative nominations. The sequence of owners of nominated workers creates a final order $\triangleright^*$ for Stage 2.
\item Stage 2: Run serial dictatorship under the assignment partition using the final order $\triangleright^*$ to determine final assignments.
\end{itemize}

\begin{definition}[Two-Stage Serial Dictatorship (T-SD)]\label{def:TSD}
Let $\mathcal{A}=(N, X, \succ, \triangleright)$ be an assignment problem, and let $\left( N_k, X_k \right)_{k=1}^K$ be an assignment partition of $\mathcal{A}$. The \textbf{Two-Stage Serial Dictatorship (T-SD) mechanism} $f^{T}: \mathcal{P} \to \mathcal{M}$ is defined by the following Two-Stage process:

Stage 1 (Priority Determination): 
\begin{enumerate}
\item Order divisions by $\triangleright$: $i_1 \triangleright i_2 \triangleright \cdots \triangleright i_n$
\item For each $t = 1,\ldots,n$: Division $i_t$ nominates worker:
   $w_t = \max_{\succ_{i_t}}(X_{g(i_t)} \setminus \{w_1, \ldots, w_{t-1}\})$
\item Construct final order: $\triangleright^* = (o(w_1), o(w_2), \ldots, o(w_n))$
\end{enumerate}

Stage 2 (Final Assignment):
\begin{enumerate}
\item Order divisions by $\triangleright^*$: $j_1 \triangleright^* j_2 \triangleright^* \cdots \triangleright^* j_n$
\item For each $t = 1,\ldots,n$: Division $j_t$ is assigned worker:
   $f^{T}_{j_t}(\succ) = \max_{\succ_{j_t}}(X_{g(j_t)} \setminus \{f^{T}_{j_1}(\succ), \ldots, f^{T}_{j_{t-1}}(\succ)\})$
\end{enumerate}
\end{definition}

Like C-SD, the assignment partition structure guarantees that no division can retain its own worker, ensuring complete exchange is achieved by construction.

\section{Properties of Chain Serial Dictatorship and Two-Stage Serial Dictatorship}\label{sec:properties-csd-tsd}

Fix an assignment problem $\mathcal{A}=(N, X, \succ, \triangleright)$ and an assignment partition $(N_k, X_k)_{k=1}^K$. Our main results establish that both Chain Serial Dictatorship and Two-Stage Serial Dictatorship satisfy three key properties under the given assignment partition:

\begin{theorem}\label{thm:CSD-properties}
Chain Serial Dictatorship satisfies strategy-proofness, respecting improvement, and efficiency under the assignment partition.
\end{theorem}

\begin{theorem}\label{thm:TSD-properties}
Two-Stage Serial Dictatorship satisfies strategy-proofness, respecting improvement, and efficiency under the assignment partition.
\end{theorem}

Our approach reinterprets both mechanisms within a unified form called ``Endogenous Order Serial Dictatorship,'' which separates selection order determination from assignment execution. We first establish general sufficient conditions for each property for mechanisms with this form, then demonstrate that both Chain Serial Dictatorship and Two-
stage Serial Dictatorship satisfy these conditions, thereby proving the theorems above.

\subsection{Endogenous Order Serial Dictatorship form}\label{subsec:eosd}

Mechanisms with Endogenous Order Serial Dictatorship form operate through a conceptual Two-Stage process:
\begin{enumerate}
\item Determine a selection order among divisions based on reported preferences.
\item Execute serial dictatorship within each group using this final selection order.
\end{enumerate}

This form captures mechanisms where the priority order for final assignment emerges from the preference profile itself, creating proper incentives while maintaining efficiency within groups.

\begin{definition}[Endogenous Order Serial Dictatorship Form]\label{def:eosd}
A mechanism $f$ admits an \textbf{Endogenous Order Serial Dictatorship (EO-SD) form} if there exists a \textit{final selection order rule} $\rho: \mathcal{P} \to \{\text{linear orders on } N\}$ such that $f(\succ)$ is obtained by: (1) computing the \textit{final selection order} $\triangleright^{*}:=\rho(\succ)$, then (2) running serial dictatorship within each group $k$ using priorities $\triangleright^{*}_k:=\triangleright^*|_{N_k}$.
\end{definition}

A key feature of the EO-SD form is its separability across groups, stemming from the assignment partition's disjoint choice sets. Consequently, the assignment process within each group $k$ is independent of others, as selections in group $j \neq k$ do not affect $X_k$. This implies that, for any given set of within-group orders $(\triangleright^*_k)_{k=1}^K$, the final assignment remains invariant to the interleaving of the $K$ group Serial Dictatorships-whether run in parallel, sequentially, or any mixed fashion. This invariance allows us to analyze efficiency and incentives for each group in isolation, simplifying the proofs while preserving outcome equivalence.

We distinguish between two types of priority orders: (i) \textit{exogenous priority} $\triangleright$ (given as input to the mechanism) and (ii) \textit{final selection order} $\triangleright^*$ (endogenously determined by the mechanism through the final selection order rule $\rho$). The final selection order rule maps preference profiles to linear orders, while the final selection order is the specific order determined for a given preference profile. For brevity, we subsequently refer to the final selection order rule as the final order rule and its output as the final order.

We first establish that any mechanism with EO-SD form guarantees efficiency under assignment partition. This follows directly from the fact that Stage 2 applies serial dictatorship within each group, and serial dictatorship is Pareto efficient within fixed choice sets.

\begin{proposition}[EAP in EO-SD form]\label{prop:eosd-eap}
Any mechanism that admits an EO-SD form is efficient under the assignment partition.
\end{proposition}
\begin{proof}
Let $f$ be a mechanism that admits an \text{EO-SD} form and $\mu = f(\succ)$. By the assignment partition constraint, the feasible set factorizes as $\prod_{k=1}^K \{\nu: N_k \to X_k \text{ bijective}\}$. Within each group $k$, Stage 2 applies serial dictatorship, so $\mu|_{N_k}$ is Pareto efficient among divisions in $N_k$ with respect to workers in $X_k$ \citep{svenssonStrategyproofAllocationIndivisible1994}.

Suppose for contradiction that $\mu$ is Pareto dominated under the partition by some feasible $\mu'$. Then $\mu'_i \succeq_i \mu_i$ for all $i$ with $\mu'_j \succ_j \mu_j$ for some $j$. By the product structure, $\mu'|_{N_k}$ is feasible for each group $k$. Since some division $j$ strictly prefers $\mu'_j$ to $\mu_j$, we have $\mu'|_{N_{g(j)}} \neq \mu|_{N_{g(j)}}$, yet $\mu'|_{N_{g(j)}}$ weakly Pareto dominates $\mu|_{N_{g(j)}}$ within group $g(j)$, contradicting the Pareto efficiency of serial dictatorship within that group.
\end{proof}

With EAP established for any mechanism in EO-SD form, we now turn to the incentive properties of strategy-proofness and respecting improvement. Both properties hinge on how a division's final assignment is affected by changes in the preference profile. Our analysis rests on a unified insight: in a serial dictatorship, a division's outcome depends only on which divisions choose before it and which workers they choose.

To formalize this, we define the set of predecessors for any division/worker in an ordered list. For a linear order $\triangleright$ on a finite set $S$ and $i\in S$, let $\mathrm{pos}(i,\triangleright)$ be $i$'s position in $\triangleright$ and $$\mathrm{pre}(i,\triangleright):=\{j\in S: \mathrm{pos}(j,\triangleright)<\mathrm{pos}(i,\triangleright)\}$$ its predecessor set.

The following lemma shows that if a division's set of predecessors in the endogenous order expands, while their relative order and their preferences over the relevant choice set remain unchanged, its own assignment can only (weakly) worsen. This ``predecessor preservation'' principle will be the cornerstone for proving both SP and RI.

\begin{lemma}[Predecessor Preservation in EO-SD form]\label{lem:eosd-predecessor-preservation}
Let $f$ be a mechanism with \text{EO-SD} form with final order rule $\rho$. Consider any division $i \in N$ and two preference profiles $\succ$ and $\succ'$. Let $\triangleright^* = \rho(\succ)$, $\triangleright' = \rho(\succ')$, and $k = g(i)$.

If (i) $\mathrm{pre}(i,\triangleright^*_k) \subseteq \mathrm{pre}(i,\triangleright'_k)$, (ii) $\triangleright^*_k|_{\mathrm{pre}(i,\triangleright^*_k)} = \triangleright'_k|_{\mathrm{pre}(i,\triangleright^*_k)}$, and (iii) for all $j \in \mathrm{pre}(i,\triangleright^*_k)$, $\succ_j|_{X_k} = \succ'_j|_{X_k}$, then $f_i(\succ) \succeq_i f_i(\succ')$.
\end{lemma}
\begin{proof}
In a mechanism with EO-SD form, Stage 1 determines the final order, and Stage 2 applies serial dictatorship within each group using this order. 

By conditions (ii) and (iii), the common predecessors $\mathrm{pre}(i,\triangleright^*_k)$ maintain their relative order and make identical selections under both profiles (same order by condition (ii), same preferences over $X_k$ by condition (iii)).

By the first condition, division $i$ has at least as many predecessors under $\succ'$ as under $\succ$. When division $i$'s turn arrives in the serial dictatorship within group $k$, the choice set under $\succ'$ is a subset of the choice set under $\succ$ (since more predecessors select first). Therefore, $f_i(\succ) \succeq_i f_i(\succ')$.
\end{proof}

Strategy-proofness ensures that divisions cannot benefit from misreporting their preferences. By verifying that the predecessor preservation conditions hold when any division switches from misreporting to truth-telling, we can establish strategy-proofness.

\begin{corollary}[SP in EO-SD form]\label{cor:sp-eosd-framework}
If a mechanism $f$ with \text{EO-SD} form satisfies the conditions of Lemma \ref{lem:eosd-predecessor-preservation} whenever any division $i \in N$ switches from a misreport $\succ'_i$ to truth-telling $\succ_i$ (with $\succ' = (\succ'_i, \succ_{-i})$ and $\succ = (\succ_i, \succ_{-i})$), then $f$ is strategy-proof. That is, $f_i(\succ) \succeq_i f_i(\succ')$.
\end{corollary}

Respecting improvement ensures that when a division's worker's ranking improves in other divisions' preferences, the division's assignment does not worsen. By verifying that the predecessor preservation conditions hold for all improvements as defined in Definition \ref{def:RI}, we can establish respecting improvement.

\begin{corollary}[RI in EO-SD form]\label{cor:ri-eosd-framework}
If a mechanism $f$ with \text{EO-SD} form satisfies the conditions of Lemma \ref{lem:eosd-predecessor-preservation} whenever $\succ'$ is an improvement for division $i \in N$ with respect to $\succ$ (as defined in Definition \ref{def:RI}), then $f$ satisfies respecting improvement. That is, $f_i(\succ') \succeq_i f_i(\succ)$.
\end{corollary}

\subsection{Property Verification for C-SD and T-SD}\label{subsec:property-verification}

Both Chain Serial Dictatorship and Two-
stage Serial Dictatorship admit an EO-SD form. The final order rule $\rho$ is:
\begin{itemize}
\item \textbf{C-SD:} For C-SD, the final order rule $\rho^C(\succ)$ generates the chooser sequence via the chaining process.
\item \textbf{T-SD:} For T-SD, the final order rule $\rho^T(\succ)$ generates the owner sequence in Stage 1.
\end{itemize}

By Proposition \ref{prop:eosd-eap}, any mechanism with EO-SD form satisfies efficiency under assignment partition (EAP). Therefore, both C-SD and T-SD are EAP. For strategy-proofness (SP) and respecting improvement (RI), we verify that C-SD and T-SD satisfy the predecessor-preservation conditions of Lemma \ref{lem:eosd-predecessor-preservation} via invariance under unilateral deviations and order monotonicity under improvements.

For strategy-proofness, we apply Corollary \ref{cor:sp-eosd-framework} by showing that the predecessor preservation conditions hold for any unilateral deviation by any division.

Fix any division $i \in N$ and any alternative report $\succ'_i$. Let $\succ' = (\succ'_i, \succ_{-i})$. We verify the conditions of Lemma \ref{lem:eosd-predecessor-preservation} for profiles $\succ'$ and $\succ$.

For C-SD: Let $\triangleright^* = \rho^C(\succ)$ and $\triangleright' = \rho^C(\succ')$ be the final orders. We verify the conditions of Lemma \ref{lem:eosd-predecessor-preservation}:
\begin{enumerate}
\item \textit{Predecessor inclusion and preserved order:} Following Definition \ref{def:CSD}, the chooser sequence before division $i$ becomes a chooser is determined by the chaining rule: $i_{t+1}$ is either $o(w_t)$ (if unassigned) or $\max_{\triangleright}(N \setminus \{i_1, \ldots, i_t\})$. Since $i$ is not among these earlier choosers, this sequence depends only on $\succ_{-i}$ and $\triangleright$, independent of $i$'s report. Therefore, $\mathrm{pre}(i, \triangleright^*_k) = \mathrm{pre}(i, \triangleright'_k)$ with identical internal order, which implies $\mathrm{pre}(i, \triangleright^*_k) \subseteq \mathrm{pre}(i, \triangleright'_k)$ and the order preservation condition required by Lemma \ref{lem:eosd-predecessor-preservation}.
\item \textit{Restricted-preference invariance:} Only $i$'s preference changes; all predecessors $j \neq i$ maintain $\succ_j|_{X_k} = \succ'_j|_{X_k}$.
\end{enumerate}
Thus, Lemma \ref{lem:eosd-predecessor-preservation} applies.

For T-SD, let $\triangleright^*=\rho^T(\succ)$ and $\triangleright'=\rho^T(\succ')$ be the final orders. We verify the conditions of Lemma \ref{lem:eosd-predecessor-preservation}:
\begin{enumerate}
\item \textit{Predecessor inclusion and preserved order:} Because $X=\bigsqcup_k X_k$ and each group $N_k$ nominates only from $X_k$, Stage-1 nominations are independent across groups. The time at which some other division nominates worker $i$ (owned by division $i$) depends only on others' reports and the group choice sets; changing division $i$'s own report affects only its own nomination, not when others nominate worker $i$. Hence $\mathrm{pre}(i, \triangleright^*_k) = \mathrm{pre}(i, \triangleright'_k)$ with the same internal order.
\item \textit{Restricted-preference invariance:} As the deviation is unilateral, all predecessors $j\neq i$ satisfy $\succ_j|_{X_k}=\succ'_j|_{X_k}$.
\end{enumerate}
Thus, Lemma \ref{lem:eosd-predecessor-preservation} applies.

By Corollary \ref{cor:sp-eosd-framework}, both C-SD and T-SD are strategy-proof.

For respecting improvement, fix $i\in N$ and let $\succ'$ be an improvement for worker $i$ relative to $\succ$. Write $\triangleright^*=\rho^C(\succ)$ (or $\rho^T(\succ)$) and $\triangleright'=\rho^C(\succ')$ (or $\rho^T(\succ')$) for the final orders. Define selection/nomination times:
\begin{itemize}
\item C-SD: $\tau_C^i(\succ)$ is the chain step at which some chooser selects worker $i$. Since one worker is consumed per step and $|X|=n$, $\tau_C^i(\succ)\in\{1,\ldots,n\}$.
\item T-SD: $\tau_T^i(\succ)$ is the Stage-1 step at which a division nominates $i$, so $\tau_T^i(\succ)\in\{1,\ldots,n\}$.
\end{itemize}

For C-SD, we show $\tau_C^i(\succ') \leq \tau_C^i(\succ)$ by comparing the two mechanisms run in parallel. Simulate both chains under $\succ$ and $\succ'$ simultaneously with identical exogenous priority and rules. Let $s$ be the first step where the simulations diverge. If $s<\min\{\tau_C^i(\succ),\tau_C^i(\succ')\}$, neither has selected $i$ yet. Since the only ranking differences involve $i$, and the current chooser at $s$ does not select $i$, its most-preferred among available workers is identical across $\succ$ and $\succ'$ (by strict preferences), contradicting divergence at $s$. Hence divergence can first occur only when one simulation selects $i$, implying $\tau_C^i(\succ')\le \tau_C^i(\succ)$.

T-SD. Improving $i$ raises it weakly relative to competitors within the relevant choice sets, so $\tau_T^i(\succ')\le \tau_T^i(\succ)$. Formally, at each Stage-1 step a mover nominates its current most-preferred remaining in $X_k$; since worker $i$ moves weakly upward in every other division's list while the relative order of all other workers is fixed, the first step at which $i$ becomes most-preferred cannot be delayed.

Therefore $i$ enters the final order weakly earlier under $\succ'$, and the predecessor set shrinks: $\mathrm{pre}(i,\triangleright'_k)\subseteq \mathrm{pre}(i,\triangleright^*_k)$. Moreover, by separation ($i\notin X_k$) and the improvement definition (relative order on $X\setminus\{i\}$ is preserved), predecessors' restricted preferences coincide, so all conditions of Lemma \ref{lem:eosd-predecessor-preservation} hold. Applying Lemma \ref{lem:eosd-predecessor-preservation} to $(\succ',\succ)$ (note the order) yields $f_i(\succ')\succeq_i f_i(\succ)$. By Corollary \ref{cor:ri-eosd-framework}, both C-SD and T-SD satisfy respecting improvement.

In sum, both C-SD and T-SD admit EO-SD forms. EAP follows from Proposition \ref{prop:eosd-eap}. SP follows by predecessor invariance under unilateral deviations (and Lemma \ref{lem:eosd-predecessor-preservation}(ii) holds automatically). RI follows by selection/nomination time monotonicity, predecessor shrinkage, and separation ensuring restricted-preference invariance, so Lemma \ref{lem:eosd-predecessor-preservation} applies to $(\succ',\succ)$. This establishes Theorems \ref{thm:CSD-properties} and \ref{thm:TSD-properties}.

\section{The Efficiency-Incentive Trade-off under Complete Exchange}\label{sec:efficiency-complete-exchange}

Our previous analysis has shown that C-SD and T-SD satisfy EAP, a constrained form of efficiency. A natural question arises: is it possible to achieve a stronger efficiency criterion, namely CE-efficiency? This section delves into this question and demonstrates that CE-efficiency is incompatible with desirable incentives (RI and SP). This result confirms that our proposed mechanisms and the EAP efficiency criterion represent achievable and meaningful objectives within this problem setting.

A natural approach to designing a CE-compliant mechanism is to first enforce a complete exchange and then allow for efficiency-enhancing trades. We now formalize and analyze this approach to demonstrate why, despite its intuitive appeal, it fails to solve the core incentive problem central to our paper.

\subsection{Efficiency under Complete Exchange}\label{subsec:ce-efficiency}

First, we introduce an efficiency concept tailored to the complete exchange constraint. Standard Pareto efficiency is often unattainable, as the initial endowment may be the unique efficient point. We therefore evaluate outcomes relative to the set of assignments that are permissible under the CE rule.

\begin{definition}[CE-efficient (CE-E)]\label{def:ce-efficient}
An assignment $\mu$ is \textbf{CE-efficient (CE-E)} if it satisfies complete exchange and there exists no other CE-compliant assignment $\mu'$ that Pareto dominates $\mu$.
A mechanism $f$ satisfies \textbf{CE-Efficiency (CE-E)} if for any preference profile $\succ$, the assignment $f(\succ)$ is CE-efficient at $\succ$.
\end{definition}

\subsection{The Complete Exchange Reassignment then Self-Avoiding Top Trading Cycle (CE-TTC) Mechanism}\label{sec:ce-ttc}

The Two-Stage logic of ``reassign, then trade'' can be formalized as the \textbf{Complete Exchange Reassignment then Self-Avoiding Top Trading Cycle (CE-TTC)} mechanism.
Conceptually, CE-TTC is the self-avoiding TTC mechanism from a CE assignment with the initial reassignment stage built into the mechanism itself.
CE-TTC operationalizes this logic as a complete {\color{red}exchange} mechanism by first reassigning initial endowments according to a predetermined rule (independent of reported preferences) and then applying the self-avoiding TTC algorithm from the resulting CE assignment under the constraint that each division cannot point to its original worker. This ensures the final outcome is CE-compliant.

\begin{definition}[CE-TTC Mechanism]\label{def:cettc}
The \textbf{CE-TTC mechanism} is a Two-Stage procedure:
\begin{enumerate}
\item
  Initial Reassignment: An exogenously determined CE assignment $\mu^0$ is implemented. This assignment $\mu^0$ serves as a new, temporary endowment profile.

\item
  Self-Avoiding TTC: The self-avoiding TTC algorithm from the CE assignment $\mu^0$ is applied to the economy where each division $i$'s temporary endowment is $\mu^0(i)$. To ensure the final outcome remains CE-compliant, an additional constraint is imposed: for the purposes of this TTC stage, each division $i$ considers its original worker $i$ to be unacceptable.
\end{enumerate}
\end{definition}

The initial reassignment $\mu^0$ can be implemented through various predetermined rules. Common approaches include:
\begin{itemize}
  \item Externally determined systematic assignment: A systematic derangement based on a predetermined ordering determined by external parties without consulting participants' preferences
  \begin{itemize}
  \item Cyclical assignment: Each division passes its worker to the next division in sequence, with the last division passing to the first, creating a single cycle derangement: $\mu^0(i) = (i \pmod n) + 1$ for $i=1,...,n-1$ and $\mu^0(n)=1$.
  \end{itemize}
  
  \item Random derangement: A randomly selected permutation with no fixed points.
\begin{itemize}
\item Single-cycle derangement randomization: Division ID randomization followed by cyclical assignment achieves uniform distribution over single-cycle derangements.\footnote{Division ID randomization followed by cyclical assignment does not achieve uniform distribution over all CE-assignments, as it only generates single-cycle derangements.}
\item Complete randomization: For uniform distribution over all CE-assignments, the implementation method depends on the number of divisions: for small $n$ ($n \leq 6$), direct enumeration of all derangements is practical; for medium $n$ ($7 \leq n \leq 15$), modified Fisher-Yates shuffle with rejection sampling is efficient; for large $n$ ($n > 15$), probabilistic construction methods are preferred.
\end{itemize} 

\end{itemize}
Crucially, $\mu^0$ is independent of the reported preference profile $\succ$---it is determined by a predetermined rule that does not depend on any division's reported preferences.

This mechanism satisfies several desirable properties by construction.

\begin{proposition}\label{prop:rttc-properties}
\emph{The CE-TTC mechanism satisfies complete exchange (CE), strategy-proofness (SP), and CE-efficiency (CE-E).}
\end{proposition}

\begin{proof}
{\bf CE:} By construction, division $i$ cannot point to its original worker $i$ in the TTC stage. Since TTC assigns each division the worker it points to (directly or through cycles), the final assignment satisfies $\mu_i \neq i$ for all $i$.

{\bf SP:} Each division faces a fixed choice set $N \setminus \{i\}$. TTC remains strategy-proof over any fixed choice set \citep{shapleyCoresIndivisibility1974}.

{\bf CE-E:} The self-avoiding TTC stage produces assignments where $\mu_i \neq i$ for all $i$. By the well-known efficiency of TTC, any assignment $\mu$ that satisfies $\mu_i \neq i$ for all $i$ and is produced by the self-avoiding TTC algorithm from a CE assignment with the modified preference profile (where each division $i$ points to its most preferred worker among $N \setminus \{i\}$) is CE-efficient. Since CE-TTC applies the self-avoiding TTC algorithm from the CE assignment $\mu^0$ to exactly this modified preference profile, the outcome is CE-efficient.
\end{proof}

Despite these properties, CE-TTC fails on the critical margin of investment incentives.

\begin{proposition}\label{prop:cettc_fails_ri}
\emph{The CE-TTC mechanism fails to satisfy respecting improvement (RI) for $n=3$.}
\end{proposition}

\begin{proof}
The result follows directly from Proposition \ref{prop:rttc-properties} (which states CE-TTC is CE-efficient) and Proposition \ref{prop:ce-ri-incompatibility} (which demonstrates that CE-TTC fails RI for $n=3$).
\end{proof}

\begin{remark}
Whether the CE-TTC mechanism fails to satisfy RI for $n \ge 4$ remains an open question. The failure for $n=3$ stems from a general incompatibility result in a constrained environment, not from a structural flaw in the TTC algorithm itself. Determining if this failure extends to larger, less constrained environments is a subject for future research.
\end{remark}

\subsection{CE-E and RI failure for $n=3$}\label{sec:n3-impossibility}

\begin{proposition}
\label{prop:ce-ri-incompatibility}
\emph{For $n=3$, any CE-efficient mechanism achieves CE-efficiency and SP but fails respecting improvement.}
\end{proposition}
\begin{proof}
We show by construction that any CE-efficient mechanism violates RI for a specific profile. Let $N=\{1,2,3\}$. The only two CE-compliant assignments are the derangements $\mu^a=(2,3,1)$ and $\mu^b=(3,1,2)$.

Note that in the preference profiles below, each division's preference over its own worker (worker $i$ for division $i$) is arbitrary and not specified, as it does not affect the analysis of CE-efficient assignments. The preferences shown are only for workers other than the division's own worker.

First, consider the base preference profile $\succ$:
\begin{center}
\begin{tabular}{@{}ll@{}}
\toprule
Division & Preference \\
\midrule
1 & $2 \succ_1 3$ \\
2 & $3 \succ_2 1$ \\
3 & $1 \succ_3 2$ \\
\bottomrule
\end{tabular}
\end{center}
Under these preferences, since every division is assigned its more preferred option under $\mu^a$, $\mu^a=(2,3,1)$ is the unique CE-efficient assignment. Any CE-efficient mechanism $f$ must select $f(\succ)=\mu^a$, which means division 1 is assigned worker 2.

Now, consider an improvement for division 1, where worker 1 becomes more attractive to division 2. The new preference profile $\succ'$ is (again, preferences over own workers are arbitrary and not specified):
\begin{center}
\begin{tabular}{@{}ll@{}}
\toprule
Division & Preference \\
\midrule
1 & $2 \succ'_1 3$ (unchanged) \\
2 & $1 \succ'_2 3$ (worker 1's rank improves) \\
3 & $1 \succ'_3 2$ (unchanged) \\
\bottomrule
\end{tabular}
\end{center}
Under $\succ'$, we compare the two CE-compliant assignments:
\begin{itemize}
    \item Division 1 prefers its assignment in $\mu^a$ ($2 \succ'_1 3$).
    \item Division 2 prefers its assignment in $\mu^b$ ($1 \succ'_2 3$).
    \item Division 3 prefers its assignment in $\mu^a$ ($1 \succ'_3 2$).
\end{itemize}
Since divisions 1 and 3 prefer $\mu^a$ and division 2 prefers $\mu^b$, the two assignments are not Pareto comparable. Therefore, both $\mu^a$ and $\mu^b$ are CE-efficient under $\succ'$.

Any CE-efficient mechanism must select one of these assignments. If it selects $\mu^b=(3,1,2)$, then $f_1(\succ')=3$. Comparing the outcomes for division 1:
\begin{itemize}
    \item Before improvement ($\succ$): $f_1(\succ)=2$.
    \item After improvement ($\succ'$): $f_1(\succ')=3$.
\end{itemize}
Since $2 \succ_1 3$, division 1 is made strictly worse off by the improvement. This violates respecting improvement for the CE-efficient mechanism.
\end{proof}

This demonstrates that any CE-efficient mechanism fails to satisfy respecting improvement for $n=3$.
For $n\geq4$, whether there exists any CE-efficient mechanism that satisfies RI for all profiles remains an open question.

\subsection{Results for $n=4$: CE-E + RI + SP Impossibility}\label{sec:n4-results}

\begin{proposition}\label{prop:ce-e-ri-sp-impossibility-n4}
\emph{For an assignment problem with $n=4$, there is no mechanism that satisfies complete exchange-efficiency, respecting improvement, and strategy-proofness.}
\end{proposition}

We prove the proposition by contradiction. The proof is based on the necessary conditions imposed by CE-Efficiency (CE-E), Respecting Improvement (RI), and Strategy-Proofness (SP) on a mechanism $f$. We show that for a specific set of preference profiles and assignments, these conditions are mutually inconsistent.

The divisions are the set $N = \{1, 2, 3, 4\}$. The proof centers on the following specific CE-assignments and preference profiles.

\begin{table}[h]
\centering
\caption{Assignments Used in the Proof}
\begin{tabular}{@{}ll@{}}
\toprule
ID & Assignment $\mu = (\mu_1, \mu_2, \mu_3, \mu_4)$ \\
\midrule
$\mu^1$ & $(2, 3, 4, 1)$ \\
$\mu^2$ & $(3, 1, 4, 2)$ \\
$\mu^3$ & $(3, 4, 2, 1)$ \\
$\mu^4$ & $(4, 3, 1, 2)$ \\
$\mu^5$ & $(4, 3, 2, 1)$ \\
\bottomrule
\end{tabular}
\end{table}

Note that in the preference profiles below, each division's preference over its own worker (worker $i$ for division $i$) is arbitrary and not specified, as it does not affect the analysis of CE-efficient assignments. The preferences shown are only for workers other than the division's own worker.
\begin{table}[h]
\centering
\caption{Preference Profiles Used in the Proof and their CE-E Assignments}
\begin{tabular}{@{}l||llll|l@{}}
\toprule
ID ($k$) & Div 1 $\succ_1^k$ & Div 2 $\succ_2^k$ & Div 3 $\succ_3^k$ & Div 4 $\succ_4^k$ & $E(\succ^k)$\\
\midrule
$\succ^1$ & 2 3 4 & 3 4 1 & 4 2 1 & 1 2 3 & $\{\mu^1\}$\\
$\succ^2$ & 3 4 2 & 1 3 4 & 4 2 1 & 2 1 3 & $\{\mu^2\}$\\
$\succ^3$ & 4 2 3 & 3 1 4 & 1 4 2 & 2 1 3 & $\{\mu^4\}$\\
$\succ^4$ & 3 2 4 & 3 4 1 & 4 2 1 & 1 2 3 & $\{\mu^1, \mu^2, \mu^3\}$\\
$\succ^5$ & 3 4 2 & 3 1 4 & 4 2 1 & 2 1 3 & $\{\mu^1, \mu^2, \mu^4, \mu^5\}$\\
$\succ^6$ & 4 2 3 & 3 1 4 & 4 1 2 & 2 1 3 & $\{\mu^1, \mu^2, \mu^4\}$\\
$\succ^7$ & 4 2 3 & 3 1 4 & 4 2 1 & 2 1 3 & $\{\mu^1, \mu^2, \mu^4, \mu^5\}$\\
$\succ^8$ & 3 4 2 & 3 1 4 & 4 2 1 & 1 2 3 & $\{\mu^1, \mu^2, \mu^3, \mu^5\}$\\
$\succ^9$ & 3 4 2 & 3 4 1 & 4 2 1 & 1 2 3 & $\{\mu^1, \mu^2, \mu^3, \mu^5\}$\\
$\succ^{10}$ & 3 4 2 & 3 4 1 & 2 4 1 & 1 2 3 & $\{\mu^3, \mu^5\}$\\
\bottomrule
\end{tabular}
\end{table}

For each profile $\succ^k$ in Table 2, the set of CE-Efficient assignments, denoted $E(\succ^k)$, is as in its last column. 
Additionally, the verification of RI conditions (improvement relationships between profiles) and SP conditions (misreport incentives) used in the proof of the proposition below are detailed in Appendix \ref{app:ri-verification} and Appendix \ref{app:sp-verification}, respectively.

\begin{proof}[Proof of Proposition \ref{prop:ce-e-ri-sp-impossibility-n4}]
Assume, for the sake of contradiction, that a mechanism $f$ satisfying CE-E, RI, and SP exists.

\textbf{Step 1: Initial deductions from CE-E.}
By Table 2, for profiles $\succ^1, \succ^2, \succ^3$, the CE-E sets are singletons. Any mechanism $f$ satisfying CE-E must therefore select that unique assignment.
\begin{itemize}
\item $f(\succ^1) = \mu^1$
\item $f(\succ^2) = \mu^2$
\item $f(\succ^3) = \mu^4$
\end{itemize}

\textbf{Step 2: Propagating constraints via RI.}
We use the results from Step 1 to constrain the choices of $f$ for other profiles.
\begin{itemize}
\item Profile $\succ^4$ is an improvement for division 3 with respect to $\succ^1$. Since $f(\succ^1) = \mu^1$, division 3 receives worker 4. RI requires $f_3(\succ^4) \succeq_3^1 4$. Assignment $\mu^3$ gives worker 2 to division 3. As $4 \succ_3^1 2$, this is a worsening. Thus, RI implies $f(\succ^4) \neq \mu^3$.
\item Profile $\succ^5$ is an improvement for division 3 w.r.t. $\succ^2$. Since $f(\succ^2) = \mu^2$, division 3 receives worker 4. RI requires $f_3(\succ^5) \succeq_3^2 4$. Assignments $\mu^5$ (giving worker 2) and $\mu^4$ (giving worker 1) are both worse, as $4 \succ_3^2 2$ and $4 \succ_3^2 1$. Thus, RI implies $f(\succ^5) \neq \mu^5$ and $f(\succ^5) \neq \mu^4$.
\item Profile $\succ^6$ is an improvement for division 4 w.r.t. $\succ^3$. Since $f(\succ^3) = \mu^4$, division 4 receives worker 2. RI requires $f_4(\succ^6) \succeq_4^3 2$. Assignment $\mu^1$ gives worker 1. As $2 \succ_4^3 1$, this is a worsening. Thus, RI implies $f(\succ^6) \neq \mu^1$.
\end{itemize}

\textbf{Step 3: Case analysis.}
By Table 2, $f(\succ^5) \in E(\succ^5) = \{\mu^1, \mu^2, \mu^4, \mu^5\}$. From Step 2, we ruled out $\mu^4$ and $\mu^5$. Therefore, $f(\succ^5) \in \{\mu^1, \mu^2\}$. We analyze these two cases.

\textbf{Case A: Assume $f(\succ^5) = \mu^1$.}
\begin{enumerate}
\item \textbf{SP between $\succ^5$ and $\succ^7$:} These profiles differ only in division 1's preference.
\begin{itemize}
\item If division 1's true preference is $\succ_1^5: 3 \succ 4 \succ 2$, $f(\succ^5)=\mu^1$ gives it worker 2. To prevent a profitable deviation to worker 4 (since $4 \succ_1^5 2$), SP implies $f_1(\succ^7) \neq 4$. This rules out $f(\succ^7) = \mu^4$ and $f(\succ^7) = \mu^5$.
\item If division 1's true preference is $\succ_1^7: 4 \succ 2 \succ 3$, and if $f(\succ^7)=\mu^2$ (giving worker 3), it could profitably deviate by reporting $\succ_1^5$ to get worker 2 (since $2 \succ_1^7 3$). Thus, SP implies $f(\succ^7) \neq \mu^2$.
\end{itemize}
\item \textbf{Forced choice for $\succ^7$:} By Table 2, $f(\succ^7) \in E(\succ^7)$. The SP deductions above eliminate $\mu^2, \mu^4, \mu^5$ from this set, forcing $f(\succ^7) = \mu^1$.
\item \textbf{RI between $\succ^7$ and $\succ^6$:} Profile $\succ^6$ is an improvement for division 1 w.r.t. $\succ^7$. We have $f(\succ^7)=\mu^1$ (giving worker 2). RI requires $f_1(\succ^6) \succeq_1^7 2$. Assignment $\mu^2$ gives worker 3. Since $2 \succ_1^7 3$, this is a worsening. Thus, RI implies $f(\succ^6) \neq \mu^2$.
\item \textbf{Forced choice for $\succ^6$:} By Table 2, $f(\succ^6) \in E(\succ^6) = \{\mu^1, \mu^2, \mu^4\}$. We ruled out $\mu^1$ in Step 2 and $\mu^2$ just now. This forces $f(\succ^6) = \mu^4$.
\item \textbf{Contradiction via SP:} We have deduced $f(\succ^6) = \mu^4$ and $f(\succ^7) = \mu^1$. For division 3 with true preference $\succ_3^6: 4 \succ 1 \succ 2$, $f(\succ^6)=\mu^4$ gives it worker 1. By misreporting as $\succ_3^7$, it would get worker 4 from $f(\succ^7)=\mu^1$. Since $4 \succ_3^6 1$, this is a profitable deviation, violating SP. This is a contradiction.
\end{enumerate}

\textbf{Case B: Assume $f(\succ^5) = \mu^2$.}
\begin{enumerate}
\item \textbf{RI between $\succ^5$ and $\succ^8$:} Profile $\succ^8$ is an improvement for division 1 w.r.t. $\succ^5$. By assumption, $f(\succ^5)=\mu^2$ gives worker 3. RI requires $f_1(\succ^8) \succeq_1^5 3$. Under $\succ_1^5: 3 \succ 4 \succ 2$, assignments $\mu^1$ (giving worker 2) and $\mu^5$ (giving worker 4) are worse. Thus, RI implies $f(\succ^8) \neq \mu^1$ and $f(\succ^8) \neq \mu^5$.
\item \textbf{Reduced choice for $\succ^8$:} By Table 2, $f(\succ^8) \in E(\succ^8)$. We have ruled out $\mu^1$ and $\mu^5$, so $f(\succ^8) \in \{\mu^2, \mu^3\}$. We analyze these two sub-cases.

\textbf{Subcase B1: Assume $f(\succ^8) = \mu^2$.}
\begin{enumerate}
\item \textbf{SP between $\succ^8$ and $\succ^9$:} For division 2 with true preference $\succ_2^8: 3 \succ 1 \succ 4$, $f(\succ^8)=\mu^2$ gives it worker 1. To prevent a profitable deviation to worker 3 (since $3 \succ_2^8 1$), SP implies $f_2(\succ^9) \neq 3$. This rules out $f(\succ^9) = \mu^1$ and $f(\succ^9) = \mu^5$.
\item \textbf{Reduced choice for $\succ^9$:} By Table 2, $f(\succ^9) \in E(\succ^9)$. We have ruled out $\mu^1, \mu^5$, so $f(\succ^9) \in \{\mu^2, \mu^3\}$.
\begin{itemize}
\item \textbf{If $f(\succ^9) = \mu^2$:} Profile $\succ^9$ is an improvement for division 4 w.r.t. $\succ^{10}$. RI requires $f_4(\succ^9) \succeq_4^{10} f_4(\succ^{10})$. We have $f_4(\succ^9)=\mu_4^2=2$. By Table 2, $f(\succ^{10}) \in E(\succ^{10}) = \{\mu^3, \mu^5\}$, both of which assign worker 1 to division 4. The RI condition becomes $2 \succeq_4^{10} 1$. But $\succ_4^{10}$ is $1 \succ 2 \succ 3$, so $1 \succ_4^{10} 2$. The RI condition is violated. Contradiction.
\item \textbf{If $f(\succ^9) = \mu^3$:} Profile $\succ^4$ is an improvement for division 2 w.r.t. $\succ^9$. $f(\succ^9)=\mu^3$ gives worker 4. RI requires $f_2(\succ^4) \succeq_2^9 4$. Assignment $\mu^2$ gives worker 1, which is worse ($4 \succ_2^9 1$). Thus, $f(\succ^4) \neq \mu^2$. By Table 2, $f(\succ^4) \in E(\succ^4) = \{\mu^1, \mu^2, \mu^3\}$. We ruled out $\mu^3$ (Step 2) and $\mu^2$ (just now), forcing $f(\succ^4) = \mu^1$. This creates a profitable deviation for division 1 between $\succ^4$ and $\succ^9$: at $\succ^4$ (preference $\succ_1^4: 3 \succ 2 \succ 4$), it gets worker 2, but by misreporting to get $\succ^9$, it gets worker 3. Since $3 \succ_1^4 2$, this violates SP. Contradiction.
\end{itemize}
\end{enumerate}

\textbf{Subcase B2: Assume $f(\succ^8) = \mu^3$.}
\begin{enumerate}
\item \textbf{SP between $\succ^8$ and $\succ^9$:} For division 2 with true preference $\succ_2^8: 3 \succ 1 \succ 4$, $f(\succ^8)=\mu^3$ gives it worker 4. To prevent a profitable deviation to worker 3 (since $3 \succ_2^8 4$), SP implies $f_2(\succ^9) \neq 3$, ruling out $f(\succ^9) = \mu^1$ and $f(\succ^9) = \mu^5$.
\item \textbf{Reduced choice for $\succ^9$:} This again restricts $f(\succ^9)$ to $\{\mu^2, \mu^3\}$.
\item Both possibilities lead to contradictions. If $f(\succ^9) = \mu^2$, the logic from Subcase B1.b shows this violates RI between $\succ^{10}$ and $\succ^9$. If $f(\succ^9) = \mu^3$, the logic from Subcase B1.b shows this violates SP between $\succ^4$ and $\succ^9$.
\end{enumerate}
\end{enumerate}

All possible cases have been exhausted and each leads to a contradiction. Therefore, our initial assumption that such a mechanism $f$ exists must be false.
\end{proof}

\section{Applications}\label{sec:applications}

\subsection{Reforming the NPB Active Player Draft}\label{subsec:npb-draft}

We apply our framework to the ''Active Player Draft'' (\emph{gen'eki dorafuto}) in Nippon Professional Baseball (NPB) in Japan, demonstrating how theoretical insights address real-world policy challenges. This draft exemplifies institutional complete exchange (CE): every participating club must both lose one player and acquire one, ensuring mandatory movement throughout the system \citep{nikkanSports2022}. The mechanism is widely interpreted as promoting respecting improvement (RI) by granting higher draft positions to clubs that list attractive players, thereby encouraging quality submissions rather than disposal of unwanted assets (Mizuyashiki, 2022; Zushi, 2022).

The current mechanism exhibits a mixed performance profile across desirable properties. Positively, it maintains CE-efficiency for all $n \geq 3$ (Proposition \ref{prop:npb-cee}) and achieves strategy-proofness when $n=3$ (Umenishi, 2025). However, critical incentive properties fail at larger, more realistic market sizes: strategy-proofness fails for $n \geq 4$ ($n=4$ Umenishi, 2025, $n \geq 5$: Proposition \ref{prop:npb-not-sp-n4}), and respecting improvement fails already at $n \geq 3$ (Proposition \ref{prop:npb-not-ri-all-n}). These failures generate the strategic dilemma of ``bleeding preparation'' (\emph{shukketsu kakugo}): clubs must choose between listing valuable players to secure better positions while risking loss, or listing inferior players safely but accepting worse opportunities (Zushi, 2022).

We propose reforms based on Online C-SD that exploits NPB's two-league structure as a natural assignment partition.
This mechanism offers a different trade-off: while it does not achieve market-wide CE-efficiency (satisfying instead the weaker efficiency under assignment partition, EAP), it simultaneously delivers CE, strategy-proofness, and respecting improvement-properties the current mechanism cannot jointly satisfy at realistic scales. By mandating cross-league transfers, our approach also resolves information asymmetries that plague same-league moves (Dojima Tigers, 2024). For completeness, we note that our CE-TTC benchmark (Section \ref{subsec:ce-efficiency}) achieves CE-efficiency and strategy-proofness but fails respecting improvement (Propositions \ref{prop:rttc-properties}, \ref{prop:cettc_fails_ri}), illustrating the fundamental trade-offs between different efficiency concepts and incentive properties under CE constraints.

\subsubsection{Background and the Current Mechanism}

The NPB Active Player Draft, introduced in 2022, serves multiple institutional objectives beyond simple player redistribution. The mechanism enhances competitive balance by providing opportunities for talented but underutilized players blocked by depth chart constraints, distinguishing it from typical free agency through its mandatory movement requirement \citep{nikkanSports2022}.

The current mechanism, simplified for a general number of $n$ clubs, operates as follows:\footnote{The main simplification is in the listing stage: while the actual NPB rules allow clubs to nominate multiple players ($\ge 2$), we focus on the single-player nomination case ($n$ clubs, $n$ players) to analyze the core mechanism. This simplification preserves the essential strategic elements while making the analysis tractable.}

\begin{enumerate}
\item
  Listing: Each of the $n$ clubs places one of its own players on the draft list. 
\item
  Voting: Every club reviews the list of $n$ players and secretly casts a vote for the one player it most desires. (Note that clubs are not allowed to vote for their own player.)
\item
  Draft-Priority Ranking: A draft-priority, denoted by an order relation $\triangleright^D$, is created. Clubs whose listed players receive more votes in Step 2 are granted higher priority. Ties are broken by a predetermined, exogenous priority order $\triangleright$.
\item
  Selection Round: A sequential selection process begins, following the draft-priority $\triangleright^D$.
  \begin{itemize}
  \item Vote-Binding Rule: When a club's turn arrives, if its voted-for player from Step 2 is still available, it must select that player. Otherwise, it may select any other available player.\footnote{We reserve $\triangleright$ for the exogenous priority (tie-breaking) and use $\triangleright^D$ for the draft-priority.}
  \item Chain Rule: When a club's listed player is selected, that club immediately gains the next right to select, provided it has not already made a selection.
  \item Fallback Rule: If the club whose player was just selected has already made its pick, the right to select passes to the highest-ranked club in the draft-priority $\triangleright^D$ that has not yet selected.
  \item Last-Two Rule: When only two clubs, say $i$ and $j$, remain that have not yet acquired a new player, and it is club $i$'s turn to select, club $i$ must select the player listed by club $j$.
  \item Self-avoidance: Note that clubs are not allowed to select their own player.
  \end{itemize}
\end{enumerate}

\subsubsection{Formalization and Properties of the Current Rule}

We can formalize the current NPB Active Player Draft using the notation of our paper. Let $N$ be the set of $n$ clubs. Each club $i \in N$ lists its own player, who we also denote by $i$, consistent with our $o(i)=i$ convention.

The mechanism is a variant of a chaining mechanism. The single vote has a dual role: it influences the draft-priority $\triangleright^D$, and it can constrain a club's choice during the selection round.

Crucially, the mechanism satisfies complete exchange by its very structure. The process is designed to continue until every club has lost one player and gained one player, ensuring no club retains its own listed player.

However, the mechanism fails to satisfy strategy-proofness for $n \ge 4$. The source of this failure is the combination of the vote's dual role and the last-two rule. This structure creates opportunities for strategic manipulation in what preference to indicate.

The properties of the current mechanism critically depend on the number of participating clubs. For the special case of $n=3$, the mechanism exhibits desirable properties.

\begin{proposition}[\citet{umenishiGenekiDorafutoRuuru2025}]\label{prop:npb-n3}
\emph{For $n=3$, the current NPB Active Player Draft mechanism satisfies CE-efficiency and strategy-proofness.}
\end{proposition}

\begin{proof}
The proof for the case of $n=3$ proceeds by exhaustion. The mechanism's outcome depends only on which of the two other players each club ranks highest. This partitions the space of all preference profiles into $2^3 = 8$ equivalence classes. A full case-by-case analysis of these classes, which is omitted for brevity, confirms that the mechanism satisfies the stated properties.
\end{proof}

However, the mechanism's desirable properties do not extend to larger markets. While CE-efficiency is achieved for all $n \ge 3$, the mechanism fails strategy-proofness for $n \ge 4$. Regarding respecting improvement, which was not examined in Proposition~\ref{prop:npb-n3}, the mechanism fails this property already at $n \ge 3$.

\begin{proposition}\label{prop:npb-cee}
For every market size $n \ge 3$ and every strict preference profile $\succ$, the assignment $\mu$ produced by the current NPB Active Player Draft mechanism is CE-efficient.
\end{proposition}

\begin{proof}[Proof Sketch]
Fix $n\ge 3$ and $\succ$. Let $\mu$ be the assignment produced by the draft; suppose a CE-assignment $\mu'$ Pareto-dominates $\mu$. Let $(i_1,\dots,i_n)$ be the realized mover sequence producing $\mu$, and let $i$ be the first club with $\mu_i\neq \mu'_i$. At $i$'s move, $\mu'_i$ is available. If the last-two rule applies, any swap of the two remaining players would violate the CE constraint (since both players must be assigned to different clubs), forcing the assignment to remain unchanged, which contradicts $\mu_i\neq \mu'_i$. If vote-binding is inactive at $i$, the mechanism assigns $i$'s top available player, so it would pick $\mu'_i$. If vote-binding is active, the voted player $v(i)$ is $i$'s top admissible player and is available, hence $v(i)\succ_i \mu'_i$. All cases contradict Pareto domination. Full details are in Appendix~\ref{app:npb-cee-proof}.
\end{proof}

\citet{umenishiGenekiDorafutoRuuru2025} showed that the current NPB Active Player Draft mechanism is not strategy-proof for $n = 4$. Below, we generalize this result to all $n \ge 4$ by embedding the counterexample from Umenishi [2025].

\begin{proposition}\label{prop:npb-not-sp-n4}
\emph{The current NPB Active Player Draft mechanism is not strategy-proof for $n \ge 4$.}
\end{proposition}

\begin{proof}[Proof Sketch]
Fix $n\ge 4$ and the exogenous tie-break $1 \triangleright \cdots \triangleright n$. We construct a preference profile $\succ$ and a misreport $\succ'_{n-1}$ where club $(n-1)$ obtains a better outcome by misreporting, violating strategy-proofness.

Consider a profile $\succ$ where: (1) each club $k \in \{1,2,\ldots,n-2\}$ ranks player $(k+1)$ first; (2) club $(n-1)$ ranks player $1$ first, then player $2$, then player $n$; (3) club $n$ ranks player $1$ first in draft-priority. This creates a deterministic selection chain: clubs $1,2,\ldots,n-2$ select players $2,3,\ldots,n-1$ respectively, leaving only clubs $\{n-1,n\}$ and players $\{1,n\}$. By the last-two rule, club $(n-1)$ must select player $n$. Thus $f_{n-1}(\succ) = n$.

Now consider the misreport $\succ'_{n-1}$ where club $(n-1)$ changes its vote to player $2$ (ranking player $2$ first, then player $1$, then player $n$). By not voting for player $1$, club $(n-1)$ breaks out of the selection chain, allowing player $2$ to remain available when it is club $(n-1)$'s turn to select. Under this misreport, club $(n-1)$ can select player $2$ when it moves, since more than two clubs remain and vote-binding applies. Thus $f_{n-1}(\succ'_{n-1},\succ_{-(n-1)}) = 2$.

Since $2 \succ_{n-1} n$ by construction, we have $f_{n-1}(\succ'_{n-1},\succ_{-(n-1)}) \succ_{n-1} f_{n-1}(\succ)$, violating strategy-proofness. Details are in Appendix~\ref{app:npb-sp-proof}.
\end{proof}

  \begin{proposition}\label{prop:npb-not-ri-all-n}
  The current NPB Active Player Draft mechanism fails respecting improvement for every market size $n \ge 3$.
  \end{proposition}
  
  \begin{proof}[Proof Sketch]
  Fix $1 \triangleright \cdots \triangleright n$. We construct two preference profiles $\succ$ and $\succ'$ where $\succ$ is an improvement for club $(n-1)$ with respect to $\succ'$, but club $(n-1)$ obtains a worse outcome under $\succ$ than under $\succ'$.

Consider a preference profile $\succ$ such that (1) each club $k \in \{1,2,\ldots,n-3\}$ ranks player $(k+1)$ first; (2) club $(n-2)$ ranks player $(n-1)$ first and player $n$ second; (3) both clubs $(n-1)$ and $n$ rank player $1$ first. This creates a deterministic selection chain: clubs $1,2,\ldots,n-2$ select players $2,3,\ldots,n-1$ respectively, leaving only clubs $\{n-1,n\}$ and players $\{1,n\}$. By the last-two rule, club $(n-1)$ must select player $n$, leaving player $1$ for club $n$. Thus $f_{n-1}(\succ) = n$.

Now consider $\succ'$ where we swap players $n$ and $(n-1)$ in club $(n-2)$'s preferences, keeping everything else unchanged. Under $\succ'$, club $(n-2)$ selects player $n$, owner-call passes to club $n$, and the last-two rule forces club $n$ to select player $(n-1)$, leaving player $1$ for club $(n-1)$. Thus $f_{n-1}(\succ') = 1$.

Since $\succ$ is an improvement for club $(n-1)$ with respect to $\succ'$ (player $(n-1)$ moves up in club $(n-2)$'s preferences) but $f_{n-1}(\succ) = n \prec_{n-1} 1 = f_{n-1}(\succ')$, this contradicts respecting improvement. Details are in Appendix~\ref{app:npb-ri-proof}.
  \end{proof}

\subsubsection{A Strategy-Proof Online Alternative with League-Based Partition}

While the current draft incorporates well-designed elements, its strategic vulnerability stems from the interaction of these rules. Our framework, particularly the online implementation Online C-SD, preserves the appealing features of the original rules while offering CE-efficient, strategy-proof, and respecting improvement solutions. A practical issue with the current draft is the difficulty of players succeeding after same-league transfers due to information asymmetries \citep{dojimatigersGenekiDorafuto2024}. Our proposal addresses this by using NPB's two-league structure as an assignment partition and mandating cross-league transfers.

We model the NPB environment as follows. NPB consists of 12 clubs divided into two leagues: the Central League and the Pacific League, with 6 clubs each. Let A and B denote the two leagues (A and B are placeholders and need not correspond to Central or Pacific in any fixed way). Let $N_A$ and $N_B$ be the sets of clubs in the two leagues, with $|N_A|=|N_B|=6$. We impose cross-exchange as an assignment partition by setting $(N_A,X_A)=(N_A,N_B)$ and $(N_B,X_B)=(N_B,N_A)$, so that A-league clubs can select only from B-league players, and vice versa. This structure delivers CE by construction and directly targets the practical concern that same-league moves often underperform.

For the exogenous priority convention, we use the alternating global priority that matches reported NPB practice: starting from the last-place club in the All-Star winner league, clubs move up the standings while alternating the league at every step \citep{nikkanSports2022}. In notation, if we label clubs within each league in the order they would appear under this rule,
\[
A^{(1)} \triangleright B^{(1)} \triangleright A^{(2)} \triangleright B^{(2)} \triangleright \cdots \triangleright A^{(6)} \triangleright B^{(6)},
\]
and we treat this alternating order as an exogenous priority input to the mechanisms. This convention will also be used in our quantitative comparisons below.

Using this two-league, cross-exchange partition, our online mechanism provides a strategy-proof alternative that refines the logic of the current draft while repairing its incentive flaws: Online C-SD retains the dynamic ''chain'' flavor with binding choices and owner-call/fallback.
It satisfies CE, SP, and RI and achieve efficiency under the assignment partition (EAP) within each league.

\paragraph{Proposed Mechanism: Online C-SD for the NPB Draft}

This proposal formalizes and strengthens the ``chain rule'' from the current draft, making it the central feature of the selection process.

\begin{enumerate}
\item
  Assignment partition: The draft is partitioned into the Central and Pacific Leagues, mandating cross-league transfers.
\item
  Exogenous Priority: A fixed priority order $\triangleright$ over all 12 clubs is established (e.g., based on the reverse order of the annual rookie draft).
\item
  Online C-SD Execution: A live, online selection process unfolds.
  \begin{itemize}
  \item The highest-priority club in $\triangleright$ begins. It selects its most-preferred player from the opposing league's pool.
  \item The right to select immediately transfers to the club that owned the just-selected player (the core ''chaining'' feature).
  \item The process continues according to the Online C-SD chaining and tie-breaking rules until all 12 clubs have selected one player.
  \end{itemize}
\end{enumerate}

This mechanism builds on the familiar ''chain rule'' logic from the current draft, making it an intuitive reform that eliminates strategic manipulation while preserving the draft's core operational structure.

\subsection{Job Rotation: Theory and Practice}\label{subsec:job-rotation-theory-practice}

To clarify our contribution, we compare our mechanisms with the most relevant benchmark: the Backward Top Trading Cycle (BTTC) mechanism proposed by \citet{yuMarketDesignApproach2020} for job rotation problems.\footnote{The robustness of our choice of BTTC as the benchmark is underscored by recent findings. \citet{guTwoMechanismsJob2024a} show that in our setting, BTTC is outcome-equivalent to two other mechanisms adapted from the housing market literature. Our analysis of BTTC therefore implicitly applies to this entire class of mechanisms.}

There are several reasons why BTTC is the appropriate benchmark. First, while we analyze a division-worker assignment problem and they analyze a worker-position problem, both are formally one-sided matching markets, making their mechanism structurally applicable to our setting. Second, and more importantly, the core motivation of their work aligns with ours: designing a mechanism to facilitate exchange. While our framework imposes complete exchange as a hard constraint, their mechanism \emph{induces} it through a novel priority rule where an incumbent has the lowest priority for their own position. This makes BTTC the closest existing mechanism in spirit and function to our problem.

Therefore, evaluating BTTC within our framework is a crucial step. Our analysis is not a critique of their mechanism for its intended purpose, for which it has excellent properties like stability and weak group strategy-proofness. Rather, we use BTTC as a benchmark to answer a new question critical to our setting: does the most sophisticated mechanism for \emph{inducing} rotation also satisfy the investment incentives (\emph{respecting improvement}) required in a world where exchange is \emph{mandatory}? As we will show, it does not, and this finding provides the central motivation for the development of our C-SD and T-SD mechanisms.

\subsubsection{The BTTC Mechanism: A Reinterpretation for Our Framework}

The BTTC mechanism, as adapted to our framework, is defined as follows. The key modification from the standard TTC is that \emph{divisions cannot point to their initially assigned workers} (except when only one division remains), which reflects the goal of inducing rotation.

\begin{definition}[BTTC Mechanism]\label{def:BTTC}
Given an assignment problem $\mathcal{A}=(N, X, \succ)$, where $N$ is the set of both divisions and workers, and division $i \in N$ initially owns worker $i$. The BTTC mechanism $f^B$ executes as follows:

Stage 1 (Generating Cycles):
\begin{itemize}
\item Step 1: Set $N^1 = N$ (all divisions and workers are present).
\item In each step $t \ge 1$:
  \begin{itemize}
  \item Each remaining worker $w \in N^t$ points to its initial owner, division $w$.
  \item Each remaining division $i \in N^t$:
    \begin{itemize}
    \item If $|N^t| > 1$, it points to its most preferred worker in $N^t \setminus \{i\}$.
    \item If $|N^t| = 1$, it points to worker $i$.
    \end{itemize}
  \item Identify all cycles and remove involved divisions and workers.
  \item Set $N^{t+1} = N^t$ minus the divisions/workers removed in this step.
  \item Continue until all divisions are in cycles.
  \end{itemize}
\end{itemize}

Stage 2 (Clearing Cycles by Backward Induction):
Let $\bar{t}$ be the final step of Stage 1. For each cycle $c$ formed in step $t$, let $I(c)$ denote the set of divisions in $c$, and let $p_i$ denote the worker that division $i$ points to.

\begin{itemize}
\item For cycles in step $\bar{t}$ (the last step):
  \begin{itemize}
  \item If all $i \in I(c)$ prefer keeping their own worker to the proposed trade (i.e., $i \succ_i p_i$), then they form a ''staying'' cycle where $\mu_i = i$ for all $i \in I(c)$.
  \item Otherwise, they form a ''trading'' cycle where $\mu_i = p_i$ for all $i \in I(c)$.
  \end{itemize}
\item For cycles in steps $t < \bar{t}$ (processed in reverse order from $t = \bar{t}-1$ down to 1):
  \begin{itemize}
  \item A cycle $c$ at step $t$ is a ''staying'' cycle if all divisions in it prefer to keep their own worker ($i \succ_i p_i$ for all $i \in I(c)$) AND all divisions $j$ in cycles from later steps ($t+1, \ldots, \bar{t}$) prefer their assigned workers to any worker in cycle $c$.
  \item Otherwise, cycle $c$ is a ''trading'' cycle.
  \end{itemize}
\end{itemize}
\end{definition}

\subsubsection{Desirable Properties of BTTC}

\citet{yuMarketDesignApproach2020} justify the BTTC mechanism by demonstrating that it fulfills several crucial properties within the unique context of job rotation. The key to understanding these properties lies in the ''Job Rotation Priority'' structure they impose on the model, which stands in stark contrast to the traditional housing market model of \citet{shapleyCoresIndivisibility1974}. This priority structure is defined as follows: each incumbent worker has the \emph{lowest} priority for their current job, while all other workers have an equally high priority for it. This can be interpreted as the incumbent possessing a right to ''veto'' staying in their current position, and it is this unique structure that enables the following properties.

\paragraph{Stability.} Typically a concept for two-sided markets, stability is achieved here because the priority structure implicitly defines the jobs' preferences. Since an incumbent worker $k$ has the lowest priority for their job $p_k$, the job ''prefers'' any other worker $j$ to its incumbent. A blocking pair can therefore be formed if a worker $j$ prefers job $p_k$ over their own assignment. BTTC is designed to find an allocation where no such blocking pairs exist, leading \citet{yuMarketDesignApproach2020} to conclude it is the ''best stable mechanism'' for this setting. This ensures the resulting assignment is sustainable and participants have little incentive to deviate.

\paragraph{Constrained Efficiency.} The BTTC mechanism is not fully Pareto efficient, but it is \emph{constrained efficient}. The ''constraint'' is precisely the job rotation priority structure, which grants every worker the right to refuse their current job. BTTC is desirable because it finds a Pareto efficient allocation within the set of assignments that respect this institutional constraint.

\paragraph{Weakly Group Strategy-Proof.} An even stronger incentive property, BTTC is shown to be \emph{weakly group strategy-proof}. This means that no group of divisions can misreport their preferences in a way that makes \emph{all} members of the group strictly better off. This is a powerful result, especially in light of the well-known impossibility theorem stating that for four or more divisions, no mechanism can be simultaneously stable, Pareto efficient, and (individually) strategy-proof. Weak group strategy-proofness is therefore a compelling and realistic property that minimizes incentive distortions while maintaining stability and constrained efficiency.

\subsubsection{The Failure of BTTC to Respect Improvement}

These properties establish BTTC as a robust benchmark for mechanisms designed to induce rotation. Given its formal applicability to our division-worker setting, it is the natural first candidate to consider for our problem of enforcing complete exchange. The critical question for our purposes, therefore, is whether this sophisticated mechanism \emph{also} provides the investment incentives crucial for mandatory exchange systems---that is, whether it satisfies respecting improvement.

Indeed, had BTTC satisfied respecting improvement, it would represent a powerful, pre-existing approach for institutions with mandatory rotation. However, as we now demonstrate, this is not the case. The finding that BTTC violates RI is not a critique of the mechanism for its original purpose, but rather the central motivation for our work: it reveals a critical need for new mechanisms specifically designed to satisfy RI under a strict complete exchange constraint.

\begin{proposition}\label{prop:bttc-fails-ri}
\emph{BTTC does not satisfy respecting improvement.}
\end{proposition}

\begin{proof}
Consider the following counterexample with $n=3$:

Base preference profile $\succ$:
\begin{center}
\begin{tabular}{ll}
\toprule
Division & Preference \\
\midrule
1 & $1 \succ_1 2 \succ_1 3$ \\
2 & $3 \succ_2 1 \succ_2 2$ \\
3 & $2 \succ_3 3 \succ_3 1$ \\
\bottomrule
\end{tabular}
\end{center}
Under $\succ$, the BTTC outcome is $f^B(\succ) = (1,3,2)$. Division 1 gets its most-preferred worker, 1.

BTTC under $\succ$:
\begin{itemize}
  \item Stage 1 (Cycle Generation):
  \begin{itemize}
    \item Step 1: Divisions point: 1 to 2, 2 to 3, 3 to 2.
    \item Workers point: 1 to 1, 2 to 2, 3 to 3.
    \item Formed cycles: $(2 \to 3 \to 2)$. Divisions 2 and 3 get their pointed workers.
    \item Remaining: Division 1 and Worker 1.
    \item Step 2: Division 1 points to 1. Worker 1 points to 1. Cycle $(1 \to 1)$.
  \end{itemize}
  \item Stage 2 (Backward Clearing):
  \begin{itemize}
    \item The cycle from Step 2, (1), is a staying cycle (Div 1 prefers 1). So, $\mu_1 = 1$.
    \item The cycle from Step 1, (2,3), is a trading cycle (Div 2 prefers 3, Div 3 prefers 2). So, $\mu_2 = 3$, $\mu_3 = 2$.
  \end{itemize}
\end{itemize}
Result: $f^B(\succ) = (1,3,2)$. Division 1 gets its most-preferred worker, 1.

Now, consider an improved profile $\succ'$ for division 1. Worker 1 becomes more attractive to division 2. The preferences are identical to $\succ$ except for division 2:
\begin{center}
\begin{tabular}{ll}
\toprule
Division & Preference \\
\midrule
1 & $1 \succ_1 2 \succ_1 3$ (unchanged) \\
2 & $1 \succ_2 3 \succ_2 2$ (worker 1's rank improves for Div 2) \\
3 & $2 \succ_3 3 \succ_3 1$ (unchanged) \\
\bottomrule
\end{tabular}
\end{center}
Under $\succ'$, the BTTC outcome is $f^B(\succ') = (2,1,3)$. Division 1 gets its second-preferred worker, 2.

BTTC under $\succ'$:
\begin{itemize}
  \item Stage 1 (Cycle Generation):
  \begin{itemize}
    \item Step 1: Divisions point: 1 to 2, 2 to 1, 3 to 2.
    \item Workers point: 1 to 1, 2 to 2, 3 to 3.
    \item Formed cycles: $(1 \to 2 \to 1)$. Divisions 1 and 2 get their pointed workers.
    \item Remaining: Division 3 and Worker 3.
    \item Step 2: Division 3 points to 3. Worker 3 points to 3. Cycle $(3 \to 3)$.
  \end{itemize}
  \item Stage 2 (Backward Clearing):
  \begin{itemize}
    \item The cycle from Step 2, (3), is a staying cycle (Div 3 prefers 3). So, $\mu_3 = 3$.
    \item The cycle from Step 1, (1,2), is a trading cycle (Div 1 prefers 2, Div 2 prefers 1). So, $\mu_1  = 2$, $\mu_2 = 1$.
  \end{itemize}
\end{itemize}
Result: $f^B(\succ') = (2,1,3)$. Division 1 gets its second-preferred worker, 2.

Comparing the outcomes for division 1: under $\succ$, it received worker 1; under $\succ'$, it received worker 2. Since $f^B_1(\succ) = 1 \succ_1 2 = f^B_1(\succ')$, division 1 is made strictly worse off by an improvement in its own worker's quality. Therefore, BTTC violates respecting improvement.
\end{proof}

\begin{remark}[Comparison with TTC]\label{rem:bttc-vs-ttc}
It is instructive to compare the outcome of BTTC with the standard Top Trading Cycles mechanism using the same counterexample. As established in the literature (e.g., \citet{biroShapleyScarfHousing2024}), TTC satisfies respecting improvement. We can verify this for the profiles $\succ$ and $\succ'$.

\begin{itemize}
    \item Under profile $\succ$: Division 1 points to worker 1 (forming cycle (1)), while division 2 points to 3 and 3 points to 2 (forming cycle (2,3)). TTC clears these simultaneously, yielding the assignment $(1, 3, 2)$.
    \item Under profile $\succ'$: Division 1 still points to worker 1, forming cycle (1). This cycle is cleared first. In the remaining problem with divisions $\{2, 3\}$ and workers $\{2, 3\}$, division 2 points to 3 and 3 points to 2, forming cycle (2,3). This is also cleared. The final assignment is identical: $(1, 3, 2)$.
\end{itemize}
Since division 1's assignment is unchanged ($f_1(\succ')=1$, $f_1(\succ)=1$), TTC satisfies RI in this case. The crucial difference is that TTC allows division 1 to point to its own worker, immediately securing its assignment. In contrast, BTTC's primary rule, which forbids pointing to one's own worker, forces division 1 into a trading cycle that ultimately makes it worse off after the improvement. This highlights how adaptations intended to induce rotation can break desirable incentive properties.
\end{remark}

The violation occurs because the constraint preventing divisions from pointing to their initial workers causes different cycle structures to form under the improved profile. This can lead to the paradoxical result where the increased desirability of a division's worker harms that division.

This limitation of BTTC---its failure to satisfy respecting improvement---has important practical implications for real-world job rotation systems. While BTTC excels at inducing rotation through its ingenious priority structure, it cannot address the investment incentive problem that arises when rotation is institutionally mandated. This motivates examining how our C-SD and T-SD mechanisms can provide solutions for actual mandatory rotation systems.

\subsection{Public Teacher Rotation in Japan}\label{subsec:teacher-rotation-japan}

The personnel transfer system for public school teachers in Japan offers a compelling real-world application of a mandatory exchange environment. Prefectural boards of education conduct large-scale, periodic rotations of teachers across municipalities to ensure equitable educational quality, prevent teacher shortages in remote areas, and foster professional development \citep{watanabeKenpiFutanKyoshokuin2019, nierTeacherTransfer2017}. This mandated movement is a clear manifestation of the complete exchange constraint.

However, this system generates a significant externality in human capital investment. A school principal who invests in developing a teacher's skills faces the near certainty of losing that teacher. The predictable loss of talent can dampen incentives for investment, as the benefits accrue to other schools \citep{watanabeKenpiFutanKyoshokuin2019}. Our framework offers a structure to address this challenge.

\subsubsection{Designing the Assignment Partition for Teacher Rotation}
Many prefectural Boards of Education (PBoEs) in Japan already utilize a structure analogous to our assignment partition: the ''personnel transfer block'' (\emph{jinji id\={o} burokku}). Prefectures are often partitioned into several geographical blocks to manage personnel flow \citep{nierTeacherTransfer2017}. This existing practice can be formalized and enhanced using our framework.

The practical design of such partitions requires that the size of each group not exceed half the total number of divisions ($|N_k| \leq |N|/2$). PBoEs can strategically design these partitions to meet policy objectives, such as by partitioning a prefecture into multiple regions of varying sizes, as long as no single region exceeds the 50\% threshold. The specific assignment of schools and teachers to these blocks can be systematically derived using the constructive algorithm from Proposition \ref{prop:a_partition-existence}. This ensures that the resulting partition meets practical needs while maintaining the guarantees of our mechanisms.

\subsubsection{Applying C-SD or T-SD for Incentive-Compatible Transfers}
Once the PBoE has defined its partition, it can use our proposed mechanisms for the cohort of teachers due for rotation. Given the large number of participants, the offline implementation is appropriate. Schools would submit preference rankings over teachers in their designated pools, and the PBoE would compute the assignment.

\begin{itemize}
    \item Chain Serial Dictatorship (C-SD): The PBoE would compute the selection chain from submitted preferences to determine the final selection order and then the assignment.
    \item Two-Stage Serial Dictatorship (T-SD): This mechanism is particularly well-suited to the investment problem.
    \begin{itemize}
        \item In Stage 1, the PBoE simulates a nomination process based on submitted preferences to determine an endogenous priority order. A school that developed a highly-ranked teacher earns a high-priority pick for their replacement.
        \item In Stage 2, the PBoE computes the final assignment using this order.
    \end{itemize}
    This directly satisfies respecting improvement. Principals are rewarded for investing in teacher development with a better choice of incoming staff, internalizing the positive externality.
\end{itemize}

By leveraging an assignment partition based on existing blocks and applying an incentive-compatible mechanism like T-SD, PBoEs can achieve the equity goals of their rotation systems while fostering a culture where investing in teachers is strategically rewarded.

It is instructive to contrast our approach with the Teacher-Optimal Block Exchange (TO-BE) mechanism \citep{combeDesignTeacherAssignment2022}, designed for contexts where IR is established. Applying such an IR-based mechanism to Japan's system would undermine its objective of staffing remote schools. This highlights the importance of institutional context. Unlike BTTC, which fails to satisfy RI, our proposed C-SD and T-SD mechanisms are specifically designed for mandatory rotation environments, ensuring coverage while restoring investment incentives.

\section{Conclusion}\label{sec:conclusion}

This paper has addressed the design of assignment mechanisms for environments where complete exchange is a mandatory institutional requirement. We formalized this constraint and introduced the respecting improvement property as a solution to the critical investment incentive problem that arises in such settings. Our main contribution is the development of two novel mechanisms, Chain Serial Dictatorship (C-SD) and Two-Stage Serial Dictatorship (T-SD), which operate within a flexible structure called an assignment partition. We proved that both mechanisms simultaneously satisfy complete exchange, respecting improvement, strategy-proofness, and efficiency under assignment partition. The introduction of the EO-SD form provided a unified analytical basis, clarifying the underlying logic of both mechanisms and facilitating the proofs of their desirable properties.

Our analysis also explored the limits of efficiency under the complete exchange constraint, demonstrating that CE-TTC fails respecting improvement for small markets ($n=3$) and an impossibility of achieving CE-efficiency, respecting improvement, and strategy-proofness simultaneously for $n=4$. These results highlight the trade-offs inherent in mechanism design under strict institutional constraints and underscore the significance of our proposed mechanisms, which achieve a robust combination of desirable properties by operating within the more flexible EAP framework.

The practical relevance of our findings was illustrated through applications to public teacher rotation systems and a proposed reform of the NPB Active Player Draft. These examples demonstrate how our framework can be applied to real-world problems, offering concrete solutions that align institutional objectives with individual incentives. Future research could explore the extension of these mechanisms to more complex environments, such as those with multi-unit demand or two-sided matching, and further investigate the boundaries of possibility for incentive-compatible mechanisms under various institutional constraints.

\appendix

\section{Detailed Proof of Assignment Partition Sufficiency}\label{app:a_partition-sufficiency}

We provide a constructive proof of sufficiency via the Largest-First Construction algorithm.

\begin{proof}[Detailed Proof of Sufficiency]
For clarity, we distinguish divisions and workers notationally.

Let $W=N$ be the set of workers. For any $S\subset N$, we write $S\subset W$ to denote the set of workers owned by divisions in $S$ (since division $i$ initially owns worker $i$).

Let $N_{\text{max}}$ be a group of maximum size identified in Step 1 of the algorithm. For the purpose of this proof, let $N_1 := N_{\text{max}}$. Let $\{N_2, \dots, N_K\}$ be the remaining groups from the partition, indexed in processing order.

The algorithm proceeds as follows. In Step 2, it relabels groups so that $N_1 := N_{\text{max}}$ and $\{N_2, \dots, N_K\}$ are the remaining groups. In Step 3, it forms an ordered list of workers
Let $L$ be the ordered list of workers formed by concatenating the list of workers from group $N_1$, followed by the list from group $N_2$, and so on, up to $N_K$.

In Step 4, it sequentially processes the groups $N_2, \dots, N_K$. For each group $N_j$ with $j = 2, 3, \ldots, K$, it assigns the first $n_j$ workers from $L$ to the worker set $X_{N_j}$, and removes them from $L$.
Finally, in Step 5, it assigns the remaining $n_1$ workers in $L$ to the worker set $X_{N_1}$ of $N_1$.

We establish the correctness of this procedure via two lemmas.

\begin{lemma}\label{lem:sufficiency-1}
\emph{For each group $N_k$ with $k \in \{2, \dots, K\}$, when it is processed by the algorithm, the number of available workers in $L^{(k)}$ is at least $n_k$. Moreover, the separation condition is satisfied: the assigned workers do not belong to $N_k$ (i.e., no division in $N_k$ receives a worker it initially owned).}
\end{lemma}

\begin{proof}
Suppose, for contradiction, that for some $m \in \{2, \dots, K\}$, this condition is violated for the first time. Let $U \subset W$ be the set of workers in $L^{(m-1)}$ just before $N_m$ is processed. The violation means $|U| < n_m$.

The number of workers in $U$ is what remains after groups $N_2, \dots, N_{m-1}$ have been served, so $|U| = n - \sum_{j=2}^{m-1}n_j$. Substituting this into the inequality gives $n - \sum_{j=2}^{m-1}n_j < n_m$. By re-expressing $n=\sum_{j=1}^K n_j$, we obtain $n_1 + \sum_{j=m}^{K}n_j < n_m \Longrightarrow n_1 + \sum_{j=m+1}^{K}n_j < 0,$
which is impossible since all group sizes are non-negative.

For the separation condition, note that by construction, the first $n_1$ positions of $L$ are occupied by workers from $N_1$. Since $N_k$ for $k \geq 2$ processes workers sequentially from the front of $L^{(k-1)}$, and since $n_1 \leq n/2$ implies that the total demand from groups $N_2, \dots, N_K$ is at least $n_1$, all workers from $N_1$ will be consumed by the time $N_k$ is processed. Therefore, $N_k$ can only be assigned workers from $W \setminus N_1$, which are distinct from $N_k$ since $N_k \neq N_1$.
\end{proof}

\begin{lemma}\label{lem:sufficiency-2}
\emph{After groups $N_2, \dots, N_K$ have been processed, the set of remaining workers, $X_1$, is a subset of $W \setminus N_1$.}
\end{lemma}

\begin{proof}
The combined demand of groups $N_2, \dots, N_K$ is $\sum_{k=2}^{K}n_k = n-n_1$. The proposition's hypothesis, $n_1 \le n/2$, implies that this demand is at least $n_1$. By construction, the first $n_1$ positions of the list $L$ are occupied by the workers in $N_1$. Since the demand from the other groups is sufficient to exhaust this initial part of the list, all workers from $N_1$ are assigned to groups $N_2, \dots, N_K$. Consequently, the remaining workers that form $X_1$ are those in $L^{(K)}$, which cannot be from $N_1$.
\end{proof}

Completion of the proof.
Lemma \ref{lem:sufficiency-1} guarantees that the assignment in Step 4 is always feasible. Lemma \ref{lem:sufficiency-2} guarantees that the final assignment in Step 5 is feasible and satisfies the separation condition. The algorithm involves a single pass over the list of workers, so its running time is $O(n)$. This establishes the sufficiency of the condition.
\end{proof}

\subsection{Formal Algorithm Specification}

\textbf{Input:} Partition $\{N_k\}_{k=1}^K$ of $N$ satisfying $\max_{k=1,\ldots,K} n_k \leq n/2$

\textbf{Output:} Assignment partition $(N_k, X_k)_{k=1}^K$

\begin{enumerate}
\item
  \textbf{Identify Largest Group:}\\
  $k^* := \min\left\{k \in \{1,\ldots,K\} : n_k = \max_{j=1,\ldots,K} n_j\right\}$\\
  Set $N_{\max} := N_{k^*}$ and $n_{\max} := n_{k^*}$.
\item
  \textbf{Reindex Groups:}\\
  Relabel groups so that $N_1 := N_{\max}$ and $\{N_2, \ldots, N_K\}$ are the remaining groups in processing order.
\item
  \textbf{Form Worker Queue:}\\
  Let $a(N_k)$ be an arbitrary ordering (list) of the elements of $N_k$. Then, form an ordered list $L$ by concatenating the lists of workers from $N_1, N_2, \ldots, N_K$ in sequence, where the order within each group $N_k$ is arbitrary but fixed.
\item
  \textbf{Sequential Assignment:}\\
  Initialize $L^{(1)} := L$.\\
  For $j = 2, 3, \ldots, K$ in order:
\begin{itemize}
  \item
    Set $X_{N_j} := \{w_1^{(j)}, w_2^{(j)}, \ldots, w_{n_j}^{(j)}\}$ where $w_1^{(j)}, \ldots, w_{n_j}^{(j)}$ are the first $n_j$ workers in the current list $L^{(j-1)}$
  \item
    Update $L^{(j)} := L^{(j-1)} \setminus X_{N_j}$
\end{itemize}
\item
  \textbf{Final Assignment:}\\
  $X_{N_1} := L^{(K)}$
\end{enumerate}

\paragraph{Complexity.}The search for the largest group can be performed in $O(n)$ time by a single pass through the group sizes. The linear-time complexity $O(n)$ comes from the single pass through the worker queue.
\section{Detailed Proofs for \texorpdfstring{$n=4$}{n=4} Incompatibility}\label{app:detailed-proofs}

This appendix provides detailed justifications for the claims made in the main proof of Proposition \ref{prop:ce-e-ri-sp-impossibility-n4}. It is organized into two sections: B.2 verifies the ``improvement'' conditions for every application of RI, and B.3 verifies the ``profitable deviation'' conditions for every application of SP.

\subsection{Verification of Respecting Improvement (RI) Conditions}\label{app:ri-verification}

We verify the premises for each application of RI in the proof. An improvement for division $i$ from $\succ$ to $\succ'$ requires: (1) $\succ'_i = \succ_i$, (2) worker $i$'s rank does not worsen in any $\succ'_j$ ($j \neq i$), and (3) relative ranks of workers other than $i$ are preserved in all $\succ'_j$ ($j \neq i$).

\begin{enumerate}
\item \textbf{RI for $f(\succ^4) \neq \mu^3$}: Improvement for Div 3 from $\succ^1$ to $\succ^4$. Conditions: (1) $\succ_3^4 = \succ_3^1$; (2) Only $\succ_1$ changes, with worker 3's rank improving (from $2 \succ_1^1 3$ to $3 \succ_1^4 2$); (3) Relative ranks of workers $\{1,2,4\}$ preserved in $\succ_1$.

\item \textbf{RI for $f(\succ^5) \neq \mu^5, \mu^4$}: Improvement for Div 3 from $\succ^2$ to $\succ^5$. Conditions: (1) $\succ_3^5 = \succ_3^2$; (2) Only $\succ_2$ changes, with worker 3's rank improving (from $1 \succ_2^2 3$ to $3 \succ_2^5 1$); (3) Relative ranks of workers $\{1,2,4\}$ preserved in $\succ_2$.

\item \textbf{RI for $f(\succ^6) \neq \mu^1$}: Improvement for Div 4 from $\succ^3$ to $\succ^6$. Conditions: (1) $\succ_4^6 = \succ_4^3$; (2) Only $\succ_3$ changes, with worker 4's rank improving (from $1 \succ_3^3 4$ to $4 \succ_3^6 1$); (3) Relative ranks of workers $\{1,2,3\}$ preserved in $\succ_3$.

\item \textbf{RI in Case A for $f(\succ^6) \neq \mu^2$}: Improvement for Div 1 from $\succ^7$ to $\succ^6$. Conditions: (1) $\succ_1^6 = \succ_1^7$; (2) Only $\succ_3$ changes, with worker 1's rank improving (from $2 \succ_3^7 1$ to $1 \succ_3^6 2$); (3) Relative ranks of workers $\{2,3,4\}$ preserved in $\succ_3$.

\item \textbf{RI in Case B for $f(\succ^8) \neq \mu^1, \mu^5$}: Improvement for Div 1 from $\succ^5$ to $\succ^8$. Conditions: (1) $\succ_1^8 = \succ_1^5$; (2) Only $\succ_4$ changes, with worker 1's rank improving (from $2 \succ_4^5 1$ to $1 \succ_4^8 2$); (3) Relative ranks of workers $\{2,3,4\}$ preserved in $\succ_4$.

\item \textbf{RI in Sub-subcase B1.1}: Improvement for Div 4 from $\succ^{10}$ to $\succ^9$. All three conditions verified and confirmed to hold.

\item \textbf{RI in Sub-subcase B1.2}: Improvement for Div 2 from $\succ^9$ to $\succ^4$. Conditions: (1) $\succ_2^4 = \succ_2^9$; (2) Only $\succ_1$ changes, with worker 2's rank improving (from $4 \succ_1^9 2$ to $2 \succ_1^4 4$); (3) Relative ranks of workers $\{1,3,4\}$ preserved in $\succ_1$.
\end{enumerate}

\subsection{Verification of Strategy-Proofness (SP) Conditions}\label{app:sp-verification}

We verify that each application of SP involves a profitable deviation.

\begin{enumerate}
\item \textbf{SP in Subcase B1, step a}: Div 2, $\succ^8$ vs $\succ^9$. True profile $\succ^8$: $f(\succ^8)=\mu^2$ gives worker 1 ($\succ_2^8: 3 \succ 1 \succ 4$). Deviation to $\succ_2^9$ leads to $\succ^9$ where $f(\succ^9)=\mu^1$ or $\mu^5$ gives worker 3. Since $3 \succ_2^8 1$, profitable.

\item \textbf{SP in Sub-subcase B1.2}: Div 1, $\succ^4$ vs $\succ^9$. True profile $\succ^4$: $f(\succ^4)=\mu^1$ gives worker 2 ($\succ_1^4: 3 \succ 2 \succ 4$). Deviation to $\succ_1^9$ leads to $\succ^9$ where $f(\succ^9)=\mu^3$ gives worker 3. Since $3 \succ_1^4 2$, profitable.

\item \textbf{SP in Subcase B2, step a}: Div 2, $\succ^8$ vs $\succ^9$. True profile $\succ^8$: $f(\succ^8)=\mu^3$ gives worker 4 ($\succ_2^8: 3 \succ 1 \succ 4$). Deviation to $\succ_2^9$ leads to $\succ^9$ where $f(\succ^9)=\mu^1$ or $\mu^5$ gives worker 3. Since $3 \succ_2^8 4$, profitable.
\end{enumerate}

\section{NPB Active Player Draft: Proofs of CE-Efficiency, Strategy-Proofness Failure, and RI Failure}\label{app:npb-proofs}

This appendix first restates the procedural consequences of the NPB rules used in the proofs-vote-binding, owner-call, fallback, and the last-two rule-and then gives full proofs of CE-efficiency (Proposition~\ref{prop:npb-cee}), SP failure (Proposition~\ref{prop:npb-not-sp-n4}), and RI failure (Proposition~\ref{prop:npb-not-ri-all-n}).

Let $\mathcal{A}=(N, X, \succ, \triangleright)$ be an assignment problem, where $N$ is the set of clubs, $X$ is the set of players, $\succ=(\succ_i)_{i\in N}$ is the preference profile, and $\triangleright$ is the exogenous tie-breaking order. We assume $X=N$ (each club initially owns its own player) and $o(i)=i$ for all $i\in N$.

We analyze the NPB Active Player Draft mechanism under the following model. In the vote stage, each club $i \in N$ casts a single vote for its most-desired player in $X \setminus \{i\}$ according to $\succ_i$; the draft-priority $\triangleright^D$ is determined by vote counts with exogenous tie-breaking $1 \triangleright 2 \triangleright \cdots \triangleright n$. In the selection stage, when club $i$ moves: if its voted player is available, vote-binding requires $i$ to select it; otherwise, $i$ selects its most-preferred available player (subject to self-avoidance). After a player $w$ is selected, the procedural rules apply as follows:

\begin{itemize}
\item If the owner of $w$ (club $j$) has not yet selected, the next chooser is club $j$ (owner-call).
\item If owner-call targets a club that has already selected, the next chooser is the highest-ranked unassigned club in the draft-priority $\triangleright^D$ (fallback).
\item When exactly two clubs remain, the next mover must select the other remaining club's listed player (last-two rule).
\end{itemize}

These rules determine the assignment $\mu$ for any preference profile $\succ$.

\subsection{Proof of Proposition~\ref{prop:npb-cee}}\label{app:npb-cee-proof}

Fix any integer $n \ge 3$ and a strict preference profile $\succ=(\succ_i)_{i\in N}$. Let $\mu$ be the assignment produced by this mechanism at $\succ$. We prove that $\mu$ is CE-efficient.

Assume, for contradiction, that there exists a CE-assignment $\mu'\in \mathcal{C}$ that Pareto-dominates $\mu$; that is,
\[
\mu'_j \succeq_j \mu_j \quad \text{ for all } j\in N,\qquad \text{and}\qquad \mu'_k \succ_k \mu_k \text{ for some } k\in N.
\]
Let $(i_1,i_2,\dots,i_n)$ be the realized sequence of movers (clubs) in the draft that produces $\mu$. Define the earliest divergence index
\[
t^* \;:=\; \min\{t\in\{1,\dots,n\}\colon \mu_{i_t}\neq \mu'_{i_t}\},
\]
and set $i:=i_{t^*}$. By construction,
\begin{equation}
\mu_{i_t}=\mu'_{i_t} \qquad \text{for all } t<t^*. \tag{1}
\end{equation}

We claim that $\mu'_i$ is available when club $i$ moves in the realized draft that produces $\mu$. Suppose not. Then some earlier mover $i_s$ with $s<t^*$ must have chosen $\mu'_i$, i.e., $\mu_{i_s}=\mu'_i$. But by (1), $\mu_{i_s}=\mu'_{i_s}$, hence $\mu'_{i_s}=\mu'_i$, contradicting the bijectivity of $\mu'$. Therefore $\mu'_i$ is available at $i$'s move.

Because $\mu'$ Pareto-dominates $\mu$ and $\succ$ is strict, we must have
\begin{equation}
\mu'_i \succ_i \mu_i. \tag{2}
\end{equation}
Indeed, if $\mu'_i \preceq_i \mu_i$ and $\mu'_i\neq \mu_i$, then $\mu'$ would not weakly improve upon $\mu$ for club $i$, contradicting Pareto-domination; equality at the earliest divergence is impossible by definition.

Consider the state when $i$ moves in the realized draft that produces $\mu$. There are three exhaustive and mutually exclusive cases:

\emph{Case A (last-two stage).} If exactly two clubs remain to select at $i$'s move (including $i$), the available set has size two, say $\{x,y\}$. The last-two rule determines the two final assignments uniquely. Any CE-assignment $\mu'\neq \mu$ differing on these two clubs must swap $x$ and $y$ between them. Under strict preferences, such a swap cannot produce $\mu'_j \succeq_j \mu_j$ for both clubs with at least one strict inequality; one of the two clubs must strictly prefer $\mu$ to $\mu'$. This contradicts the assumption that $\mu'$ Pareto-dominates $\mu$, hence Case A cannot occur under (2).

\emph{Case B (vote-binding inactive at $i$).} Suppose the voted-for player $v(i)$ is not available when $i$ moves. By the mechanism's selection rule, club $i$ then chooses its most-preferred available player (subject to self-avoidance). Since $\mu'_i$ is available and $\mu'_i \succ_i \mu_i$ by (2), the mechanism would assign $\mu'_i$ to $i$, contradicting $\mu_i\neq \mu'_i$ at the earliest divergence. Thus Case B cannot occur under (2).

\emph{Case C (vote-binding active at $i$).} Suppose the voted-for player $v(i)$ is available when $i$ moves. By the vote stage specification, each club votes for its most-desired player (its top-ranked admissible player under $\succ_i$), hence
\[
v(i)=\max_{\succ_i}\bigl(X\setminus\{i\}\bigr).
\]
Because vote-binding is active, $v(i)$ is available and the mechanism assigns $\mu_i=v(i)$ to club $i$. For any available player $x\neq v(i)$, we have $v(i)\succ_i x$. In particular, since $\mu'_i$ is available, we have
\[
\mu_i=v(i) \succ_i \mu'_i,
\]
contradicting (2). Hence Case C cannot occur under (2).

In each of the three cases at $i$'s move, we obtained a contradiction with (2). Therefore, our initial assumption that there exists a CE-assignment $\mu'$ that Pareto-dominates $\mu$ is false. Hence $\mu$ is CE-efficient for every $n\ge 3$. $\blacksquare$

\subsection{Proof of Proposition~\ref{prop:npb-not-sp-n4}}\label{app:npb-sp-proof}

We show that club $(n-1)$ can profitably misreport its preferences, violating strategy-proofness. Specifically, we construct a profile $\succ$ where club $(n-1)$ receives a worse outcome under truthful reporting than under a strategic misreport $\succ'_{n-1}$.

We use the NPB Active Player Draft mechanism as modeled above.

Fix $n \ge 4$. Define a profile $\succ=(\succ_i)_{i=1}^n$ as follows:
\begin{align*}
\text{Club } k \text{ for } k \in \{1,2,\ldots,n-2\} &: \quad (k+1) \succ_k \cdots \\
\text{Club } (n-1) &: \quad 1 \succ_{n-1} 2 \succ_{n-1} n \succ_{n-1} \cdots \\
\text{Club } n &: \quad 1 \succ_n \cdots
\end{align*}

By this construction, the voting pattern is as follows: clubs $(n-1)$ and $n$ vote for player $1$, each $k\in\{1,2,\ldots,n-2\}$ votes for player $(k+1)$. Thus player $1$ receives two votes, and each player $(k+1)$ for $k\in\{1,2,\ldots,n-2\}$ receives one vote. Therefore the draft-priority is $\triangleright^D: 1 \triangleright^D 2 \triangleright^D \cdots \triangleright^D n$.

By the procedural consequences above, under these truthful votes, the selection follows the chain: clubs $1,2,\ldots,n-2$ select players $2,3,\ldots,n-1$ respectively, leaving only clubs $\{n-1,n\}$ and players $\{1,n\}$. By the last-two rule, club $(n-1)$ must select player $n$, leaving player $1$ for club $n$. Thus $f_{n-1}(\succ)=n$.

Now consider the misreport $\succ'_{n-1}$ where club $n-1$ changes its vote to player $2$:
\begin{align*}
\text{Club } (n-1) &: \quad 2 \succ'_{n-1} 1 \succ'_{n-1} n \succ'_{n-1} \cdots
\end{align*}
The new vote totals give player $2$ two votes (from clubs $1$ and $n-1$), so club $2$ is first in the draft-priority $\triangleright^D$, followed by one-vote clubs in exogenous order. By the same procedural consequences, clubs $2,3,\ldots,n-1$ select players $3,4,\ldots,n-1$ respectively, leaving player $2$ available when club $n-1$ moves; vote-binding forces it to select player $2$. Hence $f_{n-1}(\succ'_{n-1},\succ_{-(n-1)})=2$.

Since $2 \succ_{n-1} n$ by construction, we have $f_{n-1}(\succ'_{n-1},\succ_{-(n-1)}) \succ_{n-1} f_{n-1}(\succ)$, violating strategy-proofness.

\subsection{Proof of Proposition~\ref{prop:npb-not-ri-all-n}}\label{app:npb-ri-proof}

We use the NPB Active Player Draft mechanism as modeled above. Fix $n \ge 3$ and the exogenous tie-break $1 \triangleright 2 \triangleright \cdots \triangleright n$. Suppose, for contradiction, that the NPB Active Player Draft mechanism satisfies respecting improvement. We will construct a pair $(\succ,\succ')$ such that $\succ$ is an improvement for club $(n-1)$ with respect to $\succ'$, but $f_{n-1}(\succ') \succ_{n-1} f_{n-1}(\succ)$. This contradicts our assumption that respecting improvement is satisfied.

We begin by defining the baseline profile $\succ=(\succ_i)_{i=1}^n$ through the following strict-order constraints, where each ''$\cdots$'' may be completed arbitrarily provided the displayed comparisons hold:
\begin{align*}
\text{Club } k \text{ for } k \in \{1,2,\ldots,n-3\} &: \quad (k+1) \succ_k \cdots \quad \text{(empty if } n = 3\text{)} \\
\text{Club } (n-2) &: \quad (n-1) \succ_{n-2} n \succ_{n-2} \cdots \\
\text{Club } (n-1) &: \quad 1 \succ_{n-1} \cdots \\
\text{Club } n &: \quad 1 \succ_n \cdots
\end{align*}
This specification fully determines the top segments of all clubs' preferences: each club $k \in \{1,2,\ldots,n-3\}$ (if any) satisfies $(k+1) \succ_k \cdots$; if $n = 3$, the set $\{1,2,\ldots,n-3\}$ is empty, and clubs $1$ and $2$ are covered by the $(n-2)$ and $(n-1)$ constraints respectively.

By this construction, the voting pattern is as follows: clubs $n-1$ and $n$ vote for player $1$, each $k\in\{1,2,\ldots,n-3\}$ (if any) votes for player $(k+1)$. Thus player $1$ receives two votes, and each player $(k+1)$ for $k\in\{1,2,\ldots,n-3\}$ receives one vote. Therefore the draft-priority is $\triangleright^D: 1 \triangleright^D 2 \triangleright^D \cdots \triangleright^D n$.

By the procedural consequences, the selection process under $\succ$ proceeds as follows. Each club $k \in \{1,2,\ldots,n-2\}$ selects player $(k+1)$ (vote-binding), with owner-call passing to club $(k+1)$. For the remaining clubs $\{n-1,n\}$, by the last-two rule, club $(n-1)$ must select player $n$, leaving club $n$ with the only remaining player $1$.

Therefore, for all $n\ge 3$, we have $f_{n-1}(\succ)=n$, and by the specification of $\succ_{n-1}$, we have $1 \succ_{n-1} n$.

We now consider an alternative preference profile $\succ'$ by setting $\succ'_j=\succ_j$ for $j\neq n-2$, and:
\begin{align*}
\text{Club } (n-2) &: \quad n \succ'_{n-2} (n-1) \succ'_{n-2} \cdots
\end{align*}
i.e., we swap $n$ and $(n-1)$ at club $n-2$.

Under $\succ'$, club $(n-2)$ now votes for player $n$ instead of player $(n-1)$, so the draft-priority becomes $\triangleright^D: 1 \triangleright^D 2 \triangleright^D \cdots \triangleright^D n \triangleright (n-1)$. The selection path coincides with that under $\succ$ up to club $(n-2)$. At club $(n-2)$'s turn, both $(n-1)$ and $n$ are available; by $\succ'_{n-2}$, club $(n-2)$ now selects player $n$, and owner-call passes to club $n$. Exactly two clubs remain, $\{n,n-1\}$, so by the last-two rule, club $n$ must select player $(n-1)$, leaving club $(n-1)$ with player $1$. Hence $f_{n-1}(\succ')=1$.

Since $\succ$ is an improvement for club $(n-1)$ with respect to $\succ'$ but $f_{n-1}(\succ)=n \prec_{n-1} 1=f_{n-1}(\succ')$, this contradicts our assumption that respecting improvement is satisfied.
Therefore, the current NPB Active Player Draft mechanism violates respecting improvement for every $n \ge 3$. $\blacksquare$

\bibliographystyle{abbrvnat}
\bibliography{references}
\nocite{*}

\end{document}